\documentclass[pdfa,a4paper,UKenglish,cleveref,autoref,thm-restate]{lipics-v2021}

\title{Constructing (Co)inductive Types via Large Sizes}

\author{Bastiaan Laarakker}{LIACS, Leiden University, The Netherlands}{b.g.laarakker@liacs.leidenuniv.nl}{https://orcid.org/0009-0005-3079-8505}{This work was partially supported by the Dutch Research Council (NWO) under grant number\ VI.Veni.232.286 (ChEOpS).}
\author{Dani\"el Otten}{ILLC, University of Amsterdam, The Netherlands \and \url{http://otten.co}}{daniel@otten.co}{https://orcid.org/0000-0003-2557-3959}{}
\author{Benno van den Berg}{ILLC, University of Amsterdam, The Netherlands \and \url{https://staff.fnwi.uva.nl/b.vandenberg3/}}{b.vandenberg3@uva.nl}{}{}

\funding{\emph{Benno van den Berg and Dani\"el Otten:} This publication is part of the project “The Power of Equality” (with project
number OCENW.M20.380) of the research programme Open Competition Science M 2 which is financed by the Dutch Research Council (NWO).}

\hideLIPIcs
\nolinenumbers

\relatedversion{This paper is based on the master thesis of the first author, which was supervised by the second and third authors.}
\authorrunning{B. Laarakker, D. Otten and B. van den Berg } 
\Copyright{Bastiaan Laarakker, Benno van den Berg and Dani\"el Otten} 

\acknowledgements{We are thankful to Andreas Abel, Jesper Cockx, and Andreas Nuyts for helpful discussions and ideas.}



\EventEditors{John Q. Open and Joan R. Access}
\EventNoEds{2}
\EventLongTitle{under review for 11th FSCD (International Conference on Formal Structures for Computation and Deduction)}
\EventShortTitle{FSCD 2026}
\EventAcronym{FSCD}
\EventYear{2026}
\EventDate{July 20--23, 2026}
\EventLocation{Lisbon, Portugal}
\EventLogo{}
\SeriesVolume{}
\ArticleNo{1}

\usepackage{mathtools}
\usepackage{bussproofs}
\usepackage{quiver}
\usepackage{stmaryrd}
\usepackage{wasysym}
\usepackage{capt-of}

\begin{document}
    
\renewcommand{\.}{.\,} 
\renewcommand{\:}{:} 
\newcommand{\bN}{\mathbb{N}}
\newcommand{\isdef}{\mathrel{\overset{\makebox[0pt]{\mbox{\normalfont\tiny\sffamily def}}}{=}}}
\newcommand{\brac}[1]{\langle {#1} \rangle }
\newcommand{\Type}{~\mathsf{type}}
\newcommand{\El}{\mathrm{El}\,}
\newcommand{\Size}{\mathsf{Size}}
\newcommand{\List}{\mathsf{List} \, }
\newcommand{\SList}{\mathsf{SList} \,}
\newcommand{\nil}{\mathsf{nil} \,}
\newcommand{\cons}{\mathsf{cons}}
\newcommand{\SStream}[2]{\mathsf{Stream} \, {#1} \,[{#2}] }
\newcommand{\add}[2]{\mathsf{add}_{#1} \,x \, xs}
\newcommand{\up}{\; \uparrow \!}
\newcommand{\fix}{\mathsf{fix}}
\newcommand{\F}{F}
\newcommand{\Fmap}{\mathsf{Fmap}}
\newcommand{\map}{\mathsf{map}}
\newcommand{\fold}{\mathsf{fold}}
\newcommand{\unfold}{\mathsf{unfold}}
\newcommand{\FAlg}{F\text-\sf{Alg}}
\newcommand{\sFAlg}{F[\Diamond{-}]\text-\sf{Alg}}
\newcommand{\FCoalg}{F\text-\sf{Coalg}}
\newcommand{\sFCoalg}{F[\Box{-}]\text-\sf{CoAlg}}

\newcommand{\In}{\mathsf{in}}
\newcommand{\out}{\mathsf{out}}
\newcommand{\sto}{\overset{i}{\to}}

\newcommand{\pair}{\mathsf{pair} }
\newcommand{\inl}{\mathsf{inl}\, }
\newcommand{\inr}{\mathsf{inr} \,}
\newcommand{\Id}{\mathsf{Id}}

\newcommand{\sem}[1]{\llbracket {#1} \rrbracket}

\newcommand{\cat}[1]{{\mathrm{\mathbf{#1}}}}
\newcommand{\Set}{\cat{Set}}
\newcommand{\Rel}{\cat{Rel}}
\newcommand{\Cat}{\cat{Cat}}
\newcommand{\Ptl}{\cat{Ptl}}
\newcommand{\C}{\mathcal C}
\newcommand{\E}{\mathbb E}
\newcommand{\Ob}[1]{{\mathrm{Ob}({#1})}}
\newcommand{\Hom}{\text{Hom}}
\newcommand{\id}{\mathsf{id}}

\newcommand{\fst}{\mathsf{fst}}
\newcommand{\snd}{\mathsf{snd}}

\newcommand{\ind}{\mathsf{ind}}
\newcommand{\0}{\mathbf{0}}
\newcommand{\1}{\mathbf{1}}

\newcommand{\refl}{\mathsf{refl}}

\newcommand\<\langle
\renewcommand\>\rangle
\newcommand\llb\llbracket
\newcommand\rrb\rrbracket
\newcommand\bb\mathbb
\providecommand\cal{}
\renewcommand\cal\mathcal
\renewcommand\sf\mathsf
\renewcommand\rm\mathrm
\renewcommand\frak\mathfrak

\newcommand{\cA}{\mathcal{A}}
\renewcommand{\k}{\mathtt{k}}
\newcommand{\s}{\mathtt{s}}

\newcommand{\dom}{\text{dom}}

\newcommand{\Per}{\mathsf{PER}}
\newcommand{\Perq}{\mathsf{PER_Q}}
\newcommand{\Mod}{\mathsf{Mod}}
\newcommand{\Asm}{\mathsf{Asm}}
\newcommand{\M}{\mathsf{M}}
\renewcommand{\E}{\mathsf{E}}

\newcommand{\pr}{\mathsf{pr}}
\newcommand{\Tr}[1]{{\|{#1}\|}}
\newcommand{\tr}[1]{|{#1}|}
\newcommand{\elim}{\mathsf{elim}}
\newcommand{\app}{\mathsf{app}}
\newcommand{\abs}{\mathsf{abs}}

\newcommand{\Ty}{\mathsf{Ty}}
\newcommand{\Tm}{\mathsf{Tm}}
\newcommand{\Ctxt}{\mathsf{ctx}}
\newcommand{\p}{\mathbf{p}}
\renewcommand{\v}{\mathbf{v}}
\newcommand{\U}{U}

\newcommand{\Iff}{\hspace{6pt}\text{iff}\hspace{6pt}}
\newcommand{\If}{\hspace{6pt}\text{if}\hspace{6pt}}
\newcommand{\ie}{\hspace{6pt}\text{i.e.}\hspace{6pt}}
\renewcommand{\and}{\hspace{6pt}\text{and}\hspace{6pt}}
\newcommand{\Or}{\hspace{6pt}\text{or}\hspace{6pt}}
\newcommand{\when}{\hspace{6pt}\text{when}\hspace{6pt}}

    \maketitle
    
    \begin{abstract}
        To ensure decidability and consistency of its type theory, a proof assistant should only accept terminating recursive functions and productive corecursive functions.
        Most proof assistants enforce this through syntactic conditions, which can be restrictive and non-modular.
        Sized types are a type-based alternative where (co)inductive types are annotated with additional size information.
        Well-founded induction on sizes can then be used to prove termination and productivity.

        An implementation of sized types exists in \texttt{Agda}, but it is currently inconsistent due to the addition of a largest size.
        We investigate an alternative approach, where intensional type theory is extended with a large type of sizes and parametric quantifiers over sizes.
        We show that inductive and coinductive types can be constructed in this theory, which improves on earlier work where this was only possible for the finitely-branching inductive types.
        The consistency of the theory is justified by an impredicative realisability model, which interprets the type of sizes as an uncountable ordinal.
    \end{abstract}
    
    \section{Introduction}
    \label{sec:intro}
    Dependent type theory can be viewed as both a programming language and a foundation of mathematics.
In particular, the Curry-Howard correspondence allows us to use dependently typed programming languages such as \texttt{Agda}, \texttt{Dedukti}, \texttt{Lean}, and \texttt{Rocq} as proof assistants.
To ensure consistency, these languages should only accept total functions, so they put bounds on recursive calls.
For example, if the function has an inductive domain (freely generated by constructors), then we know that the function terminates if it is only called recursively on a smaller input.
Dually, if the function has a coinductive codomain (cofreely generated by destructors), then we know that the function is productive if its output is larger than that of its recursive calls.
In both of these cases, we know that the function is total; however, totality can also be the result of a more intricate mix of induction and coinduction.

\subparagraph{Syntactic Totality.}
Totality is usually enforced via syntactic checks for termination and productivity.
By default, \texttt{Rocq} and \texttt{Lean} check for structural (co)recursion \cite{goguen_eliminating_2006,cockx_eliminating_2016,thibodeau_intensional_2020}, while \texttt{Agda} and \texttt{Dedukti} use the more permissive size-change termination principle \cite{size-change}.
Both conditions rely on a syntactic comparison of the structural size of the input/output of the original function with the structural size of the input/output of the recursive call. 
However, programs with more complex patterns of (co)recursion are not easily defined in a way that satisfies these syntactic conditions.
Moreover, syntactic methods are sensitive to slight reformulations of programs, and do not work well with higher-order programs \cite{Abel-inflationary}.

\subparagraph{Type-based Totality.}
This work is concerned with a type-based approach to enforcing totality: sized types.
Sized types were originally introduced by Hughes, Pareto, and Sabry \cite{hughes} as an alternative approach to both termination and productivity checking, in which these properties are guaranteed by the typing derivation instead of an inspection of syntax.
In this approach, both inductive and coinductive types can be annotated with additional size information.

We can think of sizes as ordinals, and they come equipped with a well-founded ordering.
In particular, for every size \(i\) and inductive type \(\sf{Ind}\) we have an ordinal approximation \(\sf{Ind}^i\), which can be seen informally as the subtype of \(\sf{Ind}\) consisting of terms that can be formed within \(i\) constructor steps.
Dually, for every size \(i\) and coinductive type \(\sf{CoInd}\) we have an ordinal approximation \(\sf{CoInd}^i\), which can be seen informally as the quotient of \(\sf{CoInd}\), where terms are identified if we cannot distinguish them within \(i\) destructor steps.

By defining functions on these annotated versions of the (co)inductive types, we introduce information about the sizes of the input/output directly in the type of the function.
This allows for greater compositionality of functions, since size-information is retained across function calls.
Moreover, it allows us to reduce both termination and productivity of recursive functions to well-founded induction on the size parameter.

Sized types, and similar approaches, have been proposed in different forms to aid  termination/productivity checking \cite{xi_sizes_2002, blanqui_sizes_2004, barthe_sizes_2004, Abel2006TypebasedTA, abel_2008}. 
The combination of sized types with dependent type theory has been studied as well \cite{Barthe, sacchini_2013, sacchini_2014, NbE, Chan_2022}.
Sized types have been implemented in \texttt{Agda} \cite{PAR-10:MiniAgda_Integrating_Sized_Dependent}, and there have been informative practical uses, including: encoding regular languages coalgebraically \cite{abel_2016}, up-to techniques for bisimilarity \cite{up-to-sized-types}, and constructing quotient-inductive types \cite{constructing-qits}.

\subparagraph*{Consistency.}
Importantly, \texttt{Agda} does not allow pattern matching on sizes: they should have no influence on computation, and are only used to aid the termination checker.
However, there are still significant limitations to both the current implementation of sized types, and to other proposals.
One design question is how the full (co)inductive type \((\sf{Co})\sf{Ind}\) should relate to its ordinal approximations \((\sf{Co})\sf{Ind}^i\).

\texttt{Agda} introduces a special largest size $\infty$, which can be thought of as a closure ordinal: an ordinal large enough that the approximation has reached the full (co)inductive type, so that we can define \((\sf{Co})\sf{Ind}\coloneqq (\sf{Co})\sf{Ind}^\infty\).
However, this approach leads to inconsistencies: \texttt{Agda} proves \(\infty<\infty\), while any strict well-order \(<\) can be proven to be irreflexive \cite{agda-issue-3026}.

We employ a different way of recovering the full (co)inductive type from its ordinal approximations: using parametric quantifiers \cite{reynolds1983types, parametric-quantifier, reldtt}.
The intuition is that the quantifiers \(\exists/\forall\) are versions of \(\Sigma/\Pi\) that restrict information: the type \(\exists i.\,A(i)\) contains dependent pairs where the first component is ignored when checking equality, while the universal quantifier \(\forall i.\,A(i)\) contains dependent functions that act uniformly on any input.
With these quantifiers, we can define the full inductive type as \(\sf{Ind}\coloneqq\exists i.\,\sf{Ind}^i\), and the full coinductive type as \(\sf{CoInd}\coloneqq\forall i.\,\sf{CoInd}^i\).
This follows to our intuition: a term of an inductive type should be bounded by some possibly infinite height (where the bound is kept abstract), whereas a term of a coinductive type should be observable up to any depth (where the depth does not influence observations). 
Using \(\Sigma/\Pi\) instead of \(\exists/\forall\) would make \((\sf{Co})\sf{Ind}\) too large, so our parametric quantifiers have the same effect as forbidding pattern matching on sizes.

To our knowledge, using parametric quantifiers for sizes was first described by Vezzosi \cite{licentiate}, although it was framed in the context of guarded type theory \cite{nakano,prod-coprogramming} to extend its applicability from coinductive to inductive types.
It was later adapted by Nuyts, Vezzosi and Devriese in work on dependent type theory with parametric quantifiers \cite{parametric-quantifier, reldtt}, where this parametric approach to sized types was described as a possible application. In their work, they also give a way to construct (co)inductive types: they construct all coinductive types but only the finitely-branching inductive types; this is because the type of sizes is interpreted as the set of natural numbers in their model.
The interpretation of sizes as a larger ordinal has been explored \cite{Abel2006TypebasedTA, Chan_2022}, but not in combination with parametricity for sizes, which is central to our construction of (co)inductive types.

\subparagraph*{Contribution.}
We show how all (co)inductive types can be encoded using sizes and parametric quantifiers.
Concretely, we extend a version of Martin-L\"of type theory --- containing an empty type, unit type, boolean type, \(\Sigma\)-types, \(\Pi\)-types, \(=\)-types, one universe, function extensionality, and no primitive (co)inductive types --- with two additional features: a type of sizes, and parametric quantifiers \(\exists\) and \(\forall\) over sizes.
In our formulation, parametricity is internalised via the addition of axioms, and quantifying over small types will produce a small type although the type of sizes is large.

We show how well-founded recursion on sizes can be used to construct ordinal approximations for any (co)inductive type.
Then, we use the parametric quantifiers to construct the full (co)inductive types: we construct both the initial algebra and the final coalgebra for any polynomial endofunctor.
Our definitions do not use the uniqueness of identity proofs (UIP) axiom, and are therefore compatible with homotopy type theory \cite{hottbook}.

We show the consistency of our theory by constructing a realisability model based on Hyland's effective topos \cite{hyland_effective_1982}: small types are interpreted as partial equivalence relations, and large types as assemblies.
This lets us interpret the type of sizes as an uncountable ordinal, which is needed to validate the axioms for the parametric quantifiers.
This models a stronger type theory: one that is also impredicative and extensional and therefore satisfies UIP.

Importantly, we do not use impredicativity and extensionality in the syntax when constructing the ordinal approximations and full (co)inductive types.
So, our constructions can be carried out in a weak theory, while our model validates a strong theory.

\subparagraph*{Related work.}
Beyond the work that we already mentioned it is worth noting the similarity between this approach to sized types and guarded type theory \cite{nakano}, as this was also the context in which this parametric approach to sizes was described \cite{licentiate}.
Guarded type theory introduces a modality $\RHD$, called the later modality.
This modality allows for the construction of guarded recursive types, in which the recursive variable only occurs under the guarding modality.
However, the guarding modality by itself is not sufficient to construct coinductive types.
Instead, the type system can be extended with so-called clock variables \cite{prod-coprogramming}, which represent different time streams and can be used to define coinductive types via universal quantification over clocks.
This construction requires quantification over clocks to be parametric, which is often expressed in the form of additional axioms \cite{mogelberg,guardedtypes16}.
This is very similar to our approach, as we encode coinductive types using parametric universal quantification over sizes. 
However, we also encode inductive types using existential quantification, whereas guarded type theory focuses on constructing coinductive types.
The relation between sized types and guarded type theory has been studied, and it was shown that guarded recursion can be simulated in a simple type theory using sized types \cite{guarded-via-sized}.

Another method for guaranteeing termination is via user-defined well-founded orderings \cite{Leroy2024}, which is possible in \texttt{Lean}, or using even weaker orderings \cite{qwf-order}.
\texttt{Rocq} also has the \texttt{equations} package \cite{hutchison_equations_2010,mangin2015equations,sozeau2019equations} that extends the termination checker at the cost of losing computational behaviour.
Although more flexible, these methods require significant manual work from the user.

\subparagraph*{Structure of the paper.}
\autoref{sec:type-theory} defines the theory in which we work, introducing the type of sizes and parametric quantification. In \autoref{sec:encoding}, we develop the encoding the approximations and for the full (co)inductive types by quantification over sizes. The consistency of the theory is justified in \autoref{sec:model}, which describes the main features of the realisability model. Finally, a summary and perspective on future work are given in \autoref{sec:conclusion}.
    
    \section{The type theory} 
    \label{sec:type-theory}
    \newenvironment{bprooftree}
{\leavevmode\hbox\bgroup}
{\DisplayProof\egroup}

We highlight some features of the dependent type theory that we work in, which is given in its entirety in \autoref{appendix:type-theory}. 
Our starting point is Martin-L\"of type theory with the empty type (\(\bot\)), unit type (\(\top\)) and boolean type (\(\sf{Bool}\)), together with dependent pair types (\(\Sigma\)), dependent function types (\(\Pi\)), and propositional equality types (\(=\)).
For \(\Sigma\) and \(\Pi\) we include definitional \(\eta\) rules.

We also include a Tarski-style universe \(U\), meaning that it comes with an explicit decoder \( \El \) from terms of \(U\) to types, which we omit when possible for readability.
The universe contains terms that act as codes for the basic types, and similar codes for \(\Pi x:A.\,B(x)\), \(\Sigma x:A.\,B(x)\) and \(a=_Aa'\) when both \(A\) and \(B(x)\) have codes in \(U\).
For each type (former), we use a superscript \((-)^U\) to distinguish the code (term of \(U\)) from the type when needed, although we will often omit this for readability.
We will call types with codes in the universe \emph{small}, whereas general types are \emph{large}.
A consequence of the $\eta$-rule for $\Pi$-types is that we may identify a family of small types $x:A\vdash B(x):\U$ indexed by the type $A$ with a function $B:A\to\U$; we will make extensive use of this in later sections.

Although we do not work with univalence, we will make use of terminology and typical constructions from homotopy type theory \cite{hottbook}.
In particular, the notion of an equivalence will be important to this work: we would like to say that a function \(f:A\to B\) is an \emph{equivalence}, written \(f:A\simeq B\), if it has a \emph{quasi-inverse}, which is a map \(g:B\to A\) such that both compositions are pointwise equal to the identity. However, we want to ensure that the type stating that \(f\) is an equivalence a \emph{mere proposition}, meaning that all terms of the type are equal. To this end, we also assume a coherence condition between the proofs of pointwise equality; see \cite[Definition 4.2.1]{hottbook}, where this definition is called a half-adjoint equivalence. We say that a type $A$ is \emph{contractible} when there exists exactly one term of type $A$ up to propositional equality.
We may also refer to terms of identity types as paths, and use operations defined by path induction, such as the action of a function \( f : A \to B \) on paths and the transport operation.

We also include the \emph{function extensionality} axiom as described in homotopy type theory \cite{hottbook}, which states that for any two dependent functions \(f,g:\Pi x:A.\,B(x)\), the canonical function \(\mathsf{happly} : (f=g)\to\Pi x:A.\,(f\,x=g\,x)\) defined by path induction is an equivalence: \[
    (f=g)\simeq \Pi{x:A} \. (f\,x=g\,x).\tag{function extensionality}
\]
We call the quasi-inverse to \(\mathsf{happly}\) obtained via this equivalence \(\mathsf{funext}\).

\subsection{Sizes and parametric quantifiers}
\label{ss:sizes}
We extend the theory with a primitive type \( \Size \), which comes with a zero term and a successor operation: \begin{center}
    \begin{bprooftree}
        \AxiomC{\(\phantom\Gamma\)}
        \UnaryInfC{$\Gamma\vdash\Size\Type$}
    \end{bprooftree} \quad
    \begin{bprooftree}
        \AxiomC{\(\phantom\Gamma\)}
        \UnaryInfC{$\Gamma\vdash 0:\Size$}
    \end{bprooftree} \quad
    \begin{bprooftree}
        \AxiomC{$\Gamma\vdash i:\Size$}
        \UnaryInfC{$\Gamma\vdash{\uparrow}i:\Size$}
    \end{bprooftree}
\end{center}
So far this looks a lot like the type of natural numbers; the difference in their behaviour will be the elimination principle, for which we first need to introduce some other notions.
In particular, we assume a preorder (also known as quasi-order) relation \(\leq\) on the type of sizes, which will be a small type.
We encode this relation as a mere proposition, which means that for any two terms \( p, q : i \leq j \) we have a propositional equality \( p = q \).
The ordering on sizes is encoded by the following rules:
\begin{center}
    \begin{bprooftree}
        \AxiomC{$\Gamma\vdash i:\Size$}
        \AxiomC{$\Gamma\vdash j:\Size$}
        \BinaryInfC{$\Gamma\vdash i\le j\Type$}
    \end{bprooftree}
    \begin{bprooftree}
        \AxiomC{$\Gamma \vdash i : \Size$}
        \UnaryInfC{$\Gamma \vdash{\sf{le}_0}\,i:0 \le i$}
    \end{bprooftree}\quad
    \begin{bprooftree}
        \AxiomC{$\Gamma \vdash i : \Size$}
        \UnaryInfC{$\Gamma \vdash{\sf{le}_{\sf{suc}}}\,i:i \le {\uparrow}i $}
    \end{bprooftree} \\[2ex]
    \begin{bprooftree}
        \AxiomC{$\Gamma \vdash i : \Size$}
        \UnaryInfC{$\Gamma \vdash\sf{le}_{\sf{refl}}\,i:i \le i$}
    \end{bprooftree}\quad
    \begin{bprooftree}
        \AxiomC{$\Gamma \vdash p:i\le j$}
        \AxiomC{$\Gamma \vdash q:j\le k$}
        \BinaryInfC{$\Gamma \vdash\sf{le}_{\sf{trans}}\,p\,q:i \le k$}
    \end{bprooftree}
\end{center}
We let $i < j$ be notation for $\uparrow \!i \le j$. 
To discuss the elimination and computation principles for sizes, we first need to add universal and existential quantifiers over sizes.
These are impredicative: although the type of sizes is large, quantifying over small types will result in a small type:
\begin{center}
    \begin{bprooftree}
        \AxiomC{$\Gamma, i : \Size \vdash A(i) : U$}
        \UnaryInfC{$\Gamma \vdash \exists i.\,A(i) : U$}
    \end{bprooftree}\quad\begin{bprooftree}
        \AxiomC{$\Gamma, i : \Size \vdash A(i) : U$}
        \UnaryInfC{$\Gamma \vdash \forall i.\,A(i) : U$}
    \end{bprooftree}
\end{center}
These quantifiers will have very similar rules to \(\Sigma\) and \(\Pi\); however, we will introduce additional axioms in \autoref{ss:equivalences} to internalise a notion of parametricity: sizes must be treated uniformly by quantifiers.
As with \(\Sigma i:\Size.\,A(i)\), terms of \(\exists i.\,A(i)\) are pairs \( \langle s, a \rangle ^\exists\) where \( s: \Size \) and \( a : A(s) \).
However, unlike \(\Sigma\), the elimination principle of \(\exists\) only allows elimination to a small type \(P\):
\begin{center}
    \begin{bprooftree}
        \AxiomC{$\Gamma\vdash s:\Size$}
        \AxiomC{$\Gamma\vdash a : A$}
        \BinaryInfC{$\Gamma\vdash\<s,a\>^\exists:\exists i. A(i)$}
    \end{bprooftree} \quad
    \begin{bprooftree}
        \AxiomC{$\Gamma, z : \exists i. A(i) \vdash P(z) : \U$}
        \noLine
        \UnaryInfC{$\Gamma, i : \Size, x : A \vdash p(i,x) :P(\langle i, x \rangle^\exists)$}
        \UnaryInfC{$\Gamma, z : \exists i. A(i) \vdash \ind_\exists(z, p(i,x)) :P(z)$}
    \end{bprooftree} \\[2ex]
    \begin{bprooftree}
        \AxiomC{}
        \UnaryInfC{$\Gamma\vdash\ind_\exists(\<s,a\>^\exists, p(i,x))\equiv p(s,a) : {P(\<s,a\>^\exists)}$}
    \end{bprooftree}
\end{center}
Consequently, we cannot define a projection to the first component, since the type \( \Size \) is not small.
Thus, the size given to the second component of a term of the existential type is inaccessible via the elimination principle and therefore kept abstract.
For \(\forall\), the rules are completely analogous to \(\Pi\):
\begin{center}
    \begin{bprooftree}
        \AxiomC{$\Gamma, i : \Size \vdash a(i) : A(i)$}
        \UnaryInfC{$\Gamma \vdash \lambda^\forall i : \Size \. a(i): \forall i . A(i) $}
    \end{bprooftree}\quad
    \begin{bprooftree}
        \AxiomC{$\Gamma \vdash f : \forall i. A(i)$}
        \AxiomC{$\Gamma \vdash s : \Size$}
        \BinaryInfC{$\Gamma \vdash f \, s : A(s)$}
    \end{bprooftree}
\end{center}
\begin{center}
    \begin{bprooftree}
        \AxiomC{}
        \UnaryInfC{$\Gamma \vdash (\lambda^\forall i \: \Size \. a(i)) \, s \equiv a(s):A(s)$}
    \end{bprooftree}\quad
    \begin{bprooftree}
        \AxiomC{}
        \UnaryInfC{$\Gamma \vdash (\lambda^\forall i \: \Size \. f \, i) \equiv f :\forall i . A(i)$}
    \end{bprooftree}
\end{center}
Since we will not allow case distinctions on sizes due to its limited elimination principle, it will not be possible to construct non-parametric functions in the syntax.

Here we have used \(\<\dots\>^\exists\) and \(\lambda^\forall\) to distinguish the pairs and functions from those in \(\Sigma\)-types and \(\Pi\)-types; however, when this does not lead to ambiguity we will drop these superscripts.
In fact, we see \(\Pi i:\Size.\,A(i) \simeq \forall i.\,A(i) \) where the maps are given by \(h\mapsto\lambda^\forall i.\,h\,i\) and \(h'\mapsto\lambda^\Pi i.\,h'\,i\).
In particular, this shows that function extensionality for \(\Pi\) implies a function extensionality principle for \(\forall\).
However, the dual principle \(\exists i.\,A(i)\simeq\Sigma i.\,A(i)\) can not be shown, since the elimination principle of \(\exists\) does not allow us to build a function to \(\Sigma\), and this equivalence is not true in the model of \autoref{sec:model}.


Another important notion will be `bounded' quantification over sizes.
Recall that for sizes \( i, j : \Size \), the type \( i < j \) is a mere proposition. We will use the notation $i : \Size \vdash \forall j < i. A(j)$ to denote the type $i : \Size \vdash \forall j. (j < i) \to A(j)$. Similarly $i : \Size \vdash \exists j < i. A(i)$ will denote the type $i : \Size \vdash \exists j. ((j < i) \times A(j))$. When defining functions in the type theory, we will often omit the inequality proofs, writing the term $\langle j, \langle p, a \rangle \rangle : \exists j < i. A(j)$ as $\langle j, a\rangle $ instead, as \( p \) is unique up to propositional equality. Similarly, if we have a term $f : \forall j < i \. A(j)$, we will simply write $f \, j$ to mean $f \, j \, p$ where $p: j < i$. In lambda notation, we may also write $\lambda j < i$ when constructing a term of type $\forall j < i\. A(j)$. 

Finally, we can give the elimination principle for the type $\Size$, which comes in the form of a fixpoint operator, corresponding to well-founded induction on sizes:
    \begin{center}
        \begin{bprooftree}
            \AxiomC{$\Gamma, i : \Size \vdash A(i) \Type$}
            \AxiomC{$\Gamma \vdash f : \forall i . ((\forall j < i. A(j)) \to A(i))$}
            \BinaryInfC{$\Gamma \vdash \fix f : \forall i. A(i)$}
        \end{bprooftree}
    \end{center}
The $\beta$-rule for $\fix$ ensures that \(\sf{fix}\,f\) behaves as a fixpoint of the map \(g\mapsto \lambda i.\,f\,i\,(\lambda j<i.\,g\,i)\):
    \begin{center}
    \begin{bprooftree}
            \AxiomC{$\Gamma \vdash f : \forall i . ((\forall j < i. A(j)) \to A(i))$}
            \UnaryInfC{$\Gamma \vdash \fix^\beta \, f : \forall i \. \fix f \, i = f \, i \, (\lambda j < i\. \fix f  \, j)$}
        \end{bprooftree}
    \end{center}
We can view this as the unfolding rule of the fixpoint operator.
The model we construct supports this \(\beta\)-rule as a definitional equality, because the model is extensional, meaning that propositional and definitional equality coincide.
However, this definitional equality would lead to a loss of decidable type checking if allowed to unfold indefinitely.
Thus, we opt to encode this as a propositional equality for now, and discuss alternatives in \autoref{sec:conclusion}.
Importantly, we can prove that, up to propositional equality, $\fix \, f$ is the \emph{unique} term satisfying this unfolding, which we may view as its propositional $\eta$-rule:
\begin{lemma}[Uniqueness of fixpoints]\label{lemma:unique-fix}
The following type is contractible for any term $f :  \forall i.(\forall j< i. A(j)) \to A(i)$:
\[
\sum_{h :\forall i .A(i)}  \forall i \. (h \, i = f \, i \, (\lambda j < i\. h \, j)).
\]
\end{lemma}
\begin{proof}
We follow a well-known strategy \cite[Theorem 5.4.7]{hottbook} \cite[Lemma 3.2]{guarded-cubical}.
First, note that the type is inhabited by the pair $(\fix f, \fix^\beta \, f)$. For any other pair $(g, g^\beta)$, we need to show $(\fix f, \fix^\beta f) = (g, g^\beta)$.
For this, it suffices to give a pointwise equality \(e:\forall i.\fix f\,i=g\,i\) and a proof \(e^\beta:(\sf{funext}(e))_*\,(\fix^\beta f)=g^\beta\).
We take \(e\coloneqq \fix(\lambda i\.\lambda e'\.p)\) where \(p\) is a concatenation of \[\hfill\begin{tikzcd}
    \fix f\,i \ar[r,"\fix^\beta f\,i",equal] & f\,i\,(\lambda j<i\.\fix f\,j) \ar[rrr,"{\mathsf{ap}_{fi}(\mathsf{funext}(\lambda j<i\.e'\,j))}",equal] &&& f\,i\,(\lambda j<i\.g\,j) \ar[r, "g^\beta\,i",equal] & g\,i.
\end{tikzcd}\hfill\]
To define \(e^\beta\), it suffices to show that the following square commutes: \[\hfill\begin{tikzcd}
    \fix f\,i \ar[r,"\fix^\beta f\,i",equal]\ar[d,"e\,i"',equal] & f\,i\,(\lambda j<i\.\fix f\,j) \ar[d,"{\mathsf{ap}_{fi}(\mathsf{funext}(\lambda j<i\.e\,j))}",equal] \\
    g\,i & f\,i\,(\lambda j<i\.g\,j). \ar[l, "g^\beta\,i",equal]
\end{tikzcd}\hfill\]
However, this is given precisely by the \(\beta\)-rule \(\fix^\beta\,(\lambda i\.\lambda e'\.p)\).
\end{proof}

\subsection{Type equivalences}
\label{ss:equivalences}
We want to internalise parametricity for quantifying over sizes in our theory, which we do by adding additional equivalences to the theory as axioms. These equivalences describe how universal and existential quantification over sizes commute with the type constructors for small types. They will allow us to encode (co)inductive types in the calculus, and their consistency will be justified by the model we construct in \ref{sec:model}. 

Before introducing the additional axioms for the calculus, we show that two equivalences can be constructed in the calculus already.
These correspond to swapping function arguments and elements of pairs. 
\begin{proposition}\label{prop:definable-eqvs}
For small types \(A\) and \(x:A,i:\sf{Size}\vdash B(x,i)\) we have the following equivalences:
\begin{align*}
    \Pi x : A. \, \forall i. B(x, i)&\simeq\forall i \. \Pi x : A.\,B(i, x), \tag{\(\exists\) parametric over \(\Sigma\)} \\
    \Sigma x : A. \, \exists i . B(x, i)&\simeq\exists i \. \Sigma x : A. \, B(i, x). \tag{\(\forall\) parametric over \(\Pi\)}
\end{align*}
\end{proposition}
We also require equivalences showing the remaining interactions between the quantifiers and \(\Sigma\), \(\Pi\), and \(<\), for all \emph{small types} \(A\), \(x:A,i:\sf{Size}\vdash B(x,i)\), and \(i: \Size \vdash C(i)\).
These are not constructible in the calculus, so we add these as axioms, similarly to what has been done in guarded type theory \cite{mogelberg, guardedtypes16}.
\begin{align*}
    \Pi x : A.\,\exists i . B(x, i)&\simeq\exists i \. \Pi x : A. \, \exists j < i .B(x, j), \tag{\(\exists\) parametric over \(\Pi\)} \\
    \Sigma x:A.\,\forall i . B(x, i) &\simeq \forall i \. \Sigma x : A.\,B(i, x), \tag{\(\forall\) parametric over \(\Sigma\)} \\
    \exists i \.  C(i)&\simeq \exists i \. \exists j < i \.C(j), \tag{\(\exists\) parametric over \(<\)} \\
    \forall i \.  C(i) &\simeq \forall i \. \forall j < i \.C(j).  \tag{\(\forall\) parametric over \(<\)}
\end{align*}
More precisely, for each equivalence, a \emph{canonical} map $f$ can already be defined in one direction: for \(\exists\) these are the right-to-left maps and for \(\forall\) these are the left-to-right maps. 
Our axioms state that each such $f$ forms an equivalence, giving us a quasi-inverse in the other direction.

Note that the first equivalence differs in the fact that the existential quantifier does not exactly commute with the small $\Pi$-type; we will call this \emph{weakly commuting} later on. The intuitive reason for this is as follows: to go from right to left, there must be a limit of all sizes in the range of the dependent product. However, we cannot guarantee that the type $B$ is monotone (also known as covariant) in $\Size$, so we make use of $\exists j < i .B(x, j)$ instead, which is monotone in $\Size$. Importantly, this axiom thus provides an implicit construction of limits of sequences of sizes in the theory, without explicitly adding a limit operator on sizes to the syntax.

With these parametricity axioms, we can also show how the quantifiers interact with fixed types such as the base types: 
\begin{proposition}
    For a small type \(A\) with \(i\notin\sf{fv}(A)\), we have the following: \begin{align*}
        \exists i.\,A\simeq A, \tag{\(\exists\) parametric over a small fixed type} \\
        \forall i.\,A\simeq A. \tag{\(\forall\) parametric over a small fixed type}
    \end{align*}
\end{proposition}
\begin{proof}
    We first prove both statements in the special case where \(A\coloneqq\top\).
    For \(\exists\), we use the fact that for any type \(B\) we have \(\bot\to B\simeq\top\), and for \(\forall\), we use function extensionality: \begin{align*}
        \exists i.\,\top
        \simeq\exists i.(\bot\to\exists j<i.\,\bot)
        \simeq\bot\to\exists i.\,\bot
        \simeq \top, &&
        \forall i.\,\top\simeq\top.
    \end{align*}
    We use this special case and the equivalence \(\Sigma x:A.\,\top\simeq A\) to prove the general versions: \begin{gather*}
        \exists i.\,A\simeq\exists i.\,\Sigma x:A.\,\top\simeq\Sigma x:A.\,\exists i.\,\top\simeq\Sigma x:A.\,\top\simeq A, \\
        \forall i.\,A\simeq\forall i.\,\Sigma x:A.\,\top\simeq\Sigma x:A.\,\forall i.\,\top\simeq\Sigma x:A.\,\top\simeq A.\qedhere
    \end{gather*}
\end{proof}

    \section{Encoding (co)inductive types}
    \label{sec:encoding}
    In this section, we encode (co)inductive types in the theory by quantifying over sizes, taking inspiration from the approach of previous work \cite{parametric-quantifier, licentiate}.
We encode these types by constructing initial algebras and terminal coalgebras for polynomial endofunctors; these functors denote the shape of a (co)inductive type.
In particular, inductive types can be described as initial algebras for polynomial endofunctors \cite{containers, homotopy-init}, while coinductive types can be described as terminal coalgebras for polynomial functors, as was done in the construction of coinductive types from inductive types \cite{ahrens_et_al:LIPIcs.TLCA.2015.17}.

\subsection{Polynomial functors}
As we will be especially interested in endofunctors on the `category' \(U\) of small types and functions, we will omit $\El$ for small types in this section. 
We will also be interested in the `category' of size-indexed small types, where the objects are functions $A : \Size \to \U$, and a morphism between $A, B : \Size \to \U$ is a map $f : \forall i \. (A \, i) \to (B \ i)$, which we will write as $f : A \overset{i}{\to} B$. 

\begin{remark*}
Because we do not assume uniqueness of identity proofs (UIP), $U$ and \(\Size\to U\) are not technically categories, but \((\infty,1)\)-categories \cite{van2011types}.
So, when we talk about functors it would also be more proper to talk about \((\infty,1)\)-functors, which satisfy additional coherence conditions.
Although we suspect that the functors we give here satisfy these coherences, we do not require these conditions for our purpose.
Thus, our internal notions corresponds to those of \emph{wild category theory} \cite{guarded-cubical}, although we will stick to the terms category and functor.
For example, an endofunctor on \(U\) (a functor from \(U\) to \(U\)) consists of the following terms:\begin{align*}
    \F : \U \to \U, \quad\text{and}\quad \F_{A,B} : (A \to B) \to (\F A \to \F B)\text{ for all }A,B:\U,
\end{align*} which satisfy \begin{align*} 
    \phi_A : \F_{A, A}(\id_A) = \id_{FA},  \quad \text{and} \quad
    \phi_{f, g} : \F_{A, C}(g \circ f) = \F_{B, C} \, g \circ \F_{A, B} f.
\end{align*} 
\end{remark*}

We briefly recall polynomial endofunctors, and how they relate to inductive types (W-types) and coinductive types (M-types), which are the types of well-founded trees and non-well-founded trees respectively.
Given a small type $A : U$ and a type $B : A \to U$, the polynomial endofunctor $P_{A, B}$ on the category of small types is defined as:
\begin{align*}
    P_{A,B}(X) &\coloneq  \Sigma{a:A} \. (B(a) \to X ), \\
    P_{A, B}(f) &\coloneq  \lambda \langle a, g\rangle \. \langle a, f \circ g \rangle.
\end{align*}
We can encode W-types and M-types as initial algebras and terminal coalgebras of polynomial endofunctors respectively.
These W-types and M-types can in turn be used to encode all inductive and coinductive types.
We give some examples of (infinitely-branching) inductive types encoded as W-types. More on W-types and M-types, their elimination principles, and  their construction can be found in, for example, \cite{containers, ahrens_et_al:LIPIcs.TLCA.2015.17}.

Given $A : U$ and $B : A \to U$, the associated W-type  \(\rm W x : A.\,B(x) \) of well-founded trees has a single constructor rule:
\begin{center}
\begin{bprooftree}
    \AxiomC{$\Gamma \vdash a : A$}
    \AxiomC{$\Gamma \vdash f : B(a) \to\rm W x : A.\,B(x) $}
    \BinaryInfC{$\Gamma \vdash \mathsf{tree}(a, f) : \rm W x : A.\, B(x) $}
\end{bprooftree}
\end{center}
Intuitively, the type $A$ determines the shapes or constructors of nodes in a well-founded tree, and $B$ determines the branching structure, i.e. arity of each constructor. The term $\mathsf{tree}(a, f)$ can be thought of as the tree with the root node labelled by $a$, and children given by $f$. For example, because the type of natural numbers has two constructors, zero and successor, it is encoded as a W-type by taking the shape to be the boolean type. The zero constructor takes no arguments and the successor constructor takes a single argument. Thus, the type $B : \mathsf{Bool} \to U$ determining the arity should defined by \( B(\mathsf{ff}) = \bot \) and \( B(\mathsf{tt}) = \top \). 

More interestingly, infinitely-branching inductive types can also be encoded as W-types by letting $B$ be infinite. One example is the inductive type $\mathsf{Brw}$ of Brouwer ordinals, which extends the constructors for the natural numbers with a limit constructor $\mathsf{lim} :(\mathbb{N} \to \mathsf{Brw}) \to \mathsf{Brw}$. Informally, assuming an existing type of natural numbers, this is encoded as the type $\rm Wx:A.\,B$ where $A\coloneqq\top+\top+\top$ is a type with three terms, which we may suggestively call $\mathsf{zero}, \mathsf{succ}$, and $\mathsf{lim}$.
We can then define $B : A \to U$ by 
\begin{align*}
    B(\mathsf{zero}) = \bot, \quad B(\mathsf{succ}) = \top, \quad \text{and} \quad B(\mathsf{lim}) = \mathbb{N}.
\end{align*}

Infinitely-branching well-founded trees are also essential in the model of constructive set theory given by Aczel \cite{ACZEL197855}. In this model, the type of all sets is given by the W-type \(W_{A : U} A\), where $U$ is the type of the constructors, and the arity is given by the identity on $U$. The idea for this type is that each set can be viewed as a well-founded tree, where its direct successors are the element of the set. 

\subsection{Size-indexed functors}
In the next sections, we will develop our theory for general functors on small types and size-indexed small types, and show how polynomial functors fit in this framework.
We will call endofunctors on the category of size-indexed small types \emph{size-indexed endofunctors}.
Following existing notation \cite{parametric-quantifier, licentiate}, we define size-indexed endofunctors $\Diamond$ and $\Box$, corresponding to bounded existential and universal quantification.
\begin{proposition}
    We have an endofunctor of the category \(\Size\to U\): 
    \begin{align*}
        &\Diamond : (\Size \to \U) \to (\Size \to \U), &
        &\Diamond_{A,B}:(A \overset{i}{\to} B) \to (\Diamond A \overset{i}{\to} \Diamond B), \\
        &\Diamond \, A  \, i \coloneq   \exists j < i\. A \, j; &
        &\Diamond_{A, B} \, f \, i \, \langle j, a \rangle \coloneq  \brac{j, f \, j \, a}.
    \end{align*}
\end{proposition}

\begin{proposition}
    We have an endofunctor of the category \(\Size\to U\):  
    \begin{align*}
        &\Box : (\Size \to \U) \to (\Size \to \U), &
        &\Box_{A,B} : (A \overset{i}{\to} B) \to (\Box A \overset{i}{\to} \Box B), \\
        &\Box \, A  \, i \coloneq  \forall j < i \. A \, j; &
        &\Box_{A , B} \, f \, i \, g \coloneq  \lambda j \. f \, j \, (g \, j).
    \end{align*}
\end{proposition}
We will write $\Diamond_i \, A$ for $\Diamond \, A \, i \coloneq  \exists j < i \. A \, j$ and similarly $\Box_i \, A$ for $\Box \,  A \, i \coloneq  \forall j < i\. A \, j$.

Given a functor $F$ on small types, we can make use of $\Diamond$ to `lift' $F$ to the category of size-indexed types by precomposing with $\Diamond$:
\begin{align*}
&\F[\Diamond -]  : (\Size \to \U) \to (\Size \to \U), &
&\F[\Diamond -]_{A,B} : (A \overset{i}{\to} B) \to ( \F[\Diamond \, A] \overset{i}{\to} \F[\Diamond \, B]), \\
&\F[\Diamond \, A] \, i \coloneq  \F( \Diamond_i \, A ) \coloneq  F (\exists j < i \. A \, j); &
&\F[\Diamond \, f]_{A, B} \, i \coloneq  F_{\Diamond_i  A, \Diamond_i  B} (\Diamond_{A, B} \, f \, i).
\intertext{Given a functor $F$ on the category of small types describing the shape of an inductive type, the size-indexed functor $\F[\Diamond -]$ should describe the shape of the size-indexed variant of the type. Analogously, \( \F[\Box -] \) is defined by precomposition with \( \Box \):}
    &\F[\Box -]  : (\Size \to \U) \to (\Size \to \U), &
    &\F[\Box -]_{A,B} : (A \overset{i}{\to} B) \to ( \F[\Box \, A] \overset{i}{\to} \F[\Box \, B]), \\
    &\F[\Box \, A] \, i \coloneq  \F( \Box_i \, A ) \coloneq  F (\forall j < i \. A \, j); & 
    &\F[\Box \, f]_{A, B} \, i \coloneq  F_{\Box_i  A, \Box_i  B}(\Box_{A, B} \, f \, i).
\end{align*}

\subsection{Inductive types}
Since we encode inductive type as initial algebras, we recall the notion of an algebra for a endofunctor on $U$ and define it for size-indexed functors.

\begin{definition}[Algebra]
    Given an endofunctor \( F : U \to U \), an \emph{F-algebra} is a pair \( (A, k_A) \) consisting of a type \( A : U \) and a structure map \( k_A : FA \to A \). A map of coalgebras \( (h, s_h) : (A, k_A) \to (B, k_B) \) is a pair of a map \( h : A \to B \) and a path \( s_h : h \circ k_A = k_B \circ Fh \). We write $\FAlg$ for the type of $F$-algebras.
An \emph{initial $F$-algebra} is an algebra \( (\mu F, \In) \), such that for any algebra \( (A, k_A) \), the following type is contractible: 
\[\sum_{h : \mu F \to A}(h \circ \In = k_A \circ Fh).\]
\end{definition}
We define \emph{size-indexed \( F \)-algebras} as algebras for \( F[\Diamond -] \).
Explicitly, a size-indexed \( F \)-algebra is a pair \( (A, k_A) \) where \( A : \Size \to U \) and \( k_A : F[\Diamond A] \sto A \).

To encode size-indexed inductive types, we need to construct initial algebras for size-indexed functors. This is done in the following lemma, where we make use of the fixpoint operator:
\begin{lemma}\label{lemma:inductive-sized-encoding}
    Let \( F : U \to U \) be a functor. The size-indexed functor \( F[\Diamond -] \) has an initial algebra \( (\mu^\Diamond F, \In^\Diamond) \), where the first component is given by:
\[
    \mu^\Diamond F \, i \coloneqq  \fix (\lambda i. \lambda  X. F(\Diamond_i X)) \, i.
\] 
\end{lemma}
\begin{proof}[Proof Sketch]
    We note that the propositional \( \beta \)-rule for the fixpoint operator gives us the following equality:
    \begin{align*}
        \mu^\Diamond F \, i &\coloneqq\fix (\lambda i \. \lambda X\. F(\Diamond_i X)) \, i \\
        &= (\lambda i \. \lambda X\. F(\Diamond_i X)) \, i \, (\lambda r < s \. \fix (\dots) \, r) \\
        &\equiv F(\Diamond_i \, \fix (\lambda s \. \lambda X\. F(\Diamond_s X)) ) \\
        & =: F [\Diamond (\mu^\Diamond F)] \, i
    .\end{align*}
    The structure map \( \In^\Diamond : F[\Diamond (\mu^\Diamond F)] \sto \mu^\Diamond F \) is then given by the pointwise transport operation. Given an algebra \( (A, k_A) \), the term \( \fold^\Diamond \, A \, k_A : \mu^\Diamond F \sto A\) is also defined using the fixpoint operator, and can be shown to be an algebra morphism. Initiality then follows from the contractibility of \( \fix \!\) described in \autoref{lemma:unique-fix}, together with function extensionality.
\end{proof}
Thus, we have a way of constructing size-indexed types, which can be thought of as the ordinal approximations for our inductive types. But our main interest is in showing that $F$ itself has an initial algebra that can be constructed via existential quantification over the initial \( F[\Diamond{-}] \)-algebra.
To this end, we first note the following properties:
\begin{lemma}
    Every \( F \)-algebra \( (A, k_A) \) induces an \( F[\Diamond{-}] \)-algebra \( (TA, k_{TA}) \), where the first component is 
    \(
        TA \coloneqq (\lambda i. A) : \Size \to U.
    \) 
This extends to a functor \( T : \FAlg \to \sFAlg \). 
\end{lemma}
\begin{proof}[Proof Sketch]
    We require a structure map \( k_{TA} : F[\Diamond TA] \sto TA \). Noting that \( i \notin \text{fv}(A) \), we can define the following auxiliary term:
    \begin{align*}
        & \mathsf{extract} : \Pi A:U \. \Diamond T A \sto TA, \\
        & \mathsf{extract}_A \, i \, \langle j, a \rangle \coloneqq  a
    .\end{align*}
    This is well-defined because $TA \, j \equiv A$. As the name suggests, it allows us to extract a term away from the irrelevant size-index.
    This allows us to define \( k_{TA} \, i \coloneqq k_A \circ F(\mathsf{extract}_A \, i)\), since $F(\mathsf{extract}_A \, i) : F[\Diamond_i TA] \to FA$.
    
    Given a algebra map \( f : (A, k_A) \to (B, k_B) \), we lift this to a map between size-indexed algebras \( T\,f : (TA, k_{TA}) \to (TB, k_{TB}) \) defined by \( Tf \coloneqq \lambda i . f \). The fact that this forms an algebra map follows from the fact that $\mathsf{extract}_A$ is natural in $A$.
\end{proof}

We can also define a functor from size-indexed algebras to algebras by existential quantification on the carrier of the size-indexed algebra.
However, this construction only works on the class of functors satisfying the following condition. As we will see, the equivalences we axiomatised ensure that this includes all polynomial functors.

\begin{definition}\label{def:weakly-commute-exist}
    Let $\F$ be an endofunctor on the category of small types. We say that $\F$ \emph{weakly commutes} with the existential type if the following canonical map is an equivalence for any $A : \Size \to \U$:
    \begin{align*}
    \mathsf{can}_A &: \exists i \. F [\Diamond A] \, i \to \F (\exists i .A \, i), & \varphi_A& : \forall i \. (\exists j < i . A \, j \to \exists j . A \, j), \\
    \mathsf{can}_A & \, \langle i, a' \rangle \coloneq  F \, (\varphi_A \, i) \, a' & \varphi_A& \, i \, \brac{j, a} \coloneq  \brac{j, a}.
    \end{align*}
\end{definition}

For functors that satisfy this property, we can define a functorial map $\exists : F[\Diamond{-}]\text-\sf{Alg} \to F\text-\sf{Alg}$ by lifting the existential quantifier to work on algebras.
\begin{proposition}
    Given a functor $F$ on the category of small types that weakly commutes with the existential quantifier, a size-indexed $F$-algebra $(A, k_A)$ gives rise to an $F$-algebra $(\exists i . A \, i, k_{\exists A})$. Moreover, this operation is functorial.
\end{proposition}
\begin{proof}[Proof Sketch]
    We first note that for any map $f : A \sto B$, we can define a map between existentially quantified types as follows:
    \begin{align*}
        \exists f : \exists i . A \, i \to \exists i . B \, i, \\
        \exists f \, \langle i, a \rangle \coloneq  \langle i, f \, i \, a \rangle.
    \end{align*}

    Let $\mathsf{can}^{-1}_A : F(\exists i . A \, i) \to (\exists i . F [\Diamond A] \, i) $ be the equivalence obtained from $F$ weakly commuting with the existential type. We define the map $k_{\exists A} : F (\exists i . A \, i) \to \exists i. A \, i$ by
    \begin{align*}
    k_{\exists A}  \coloneq  \exists \, k_A \circ \mathsf{can}^{-1}_A.
    \end{align*}
    Given an algebra map $f :(A, k_A) \to (B, k_B)$, we define an algebra map $(\exists i . A \, i, k_{\exists A}) \to (\exists i . B \, i, k_{\exists B})$ as $\exists f : \exists i . A \, i \to \exists i . B \, i$.
    The fact that this constitutes an algebra morphism follows from $\mathsf{can}_A$ being natural in $A$. Thus, we have a functor $\exists : \sFAlg \to \FAlg$.
\end{proof}

Finally, we are able to prove our desired result, namely that for functors that weakly commute with the existential type, their initial algebra is obtained by existential quantification over the initial size-indexed algebra.
\begin{theorem}\label{theorem:inductive-construction}
    Let $\F$ be a functor on the category of small types that weakly commutes with the existential type. Then $\F$ has an initial algebra with first component:
    \[
    \mu  \F \coloneq  \exists i . (\mu^\Diamond F) \, i.
    \]
\end{theorem}
\begin{proof}[Proof Sketch]
    The categorical intuition is that this follows from the fact that the functorial map $\exists : \sFAlg \to \FAlg$ of the previous proposition is the left adjoint of $T : \FAlg \to \sFAlg$.
    Since left adjoints preserve initial objects, we find that $\exists i. (\mu^\Diamond F) \, i$ is the initial algebra of $F$, with the structure map $\In \coloneq  \exists \, \In^\Diamond \circ \mathsf{can}^{-1}_{\mu^\Diamond F}$, which is also an equivalence.

    We can work this out in more detail: for any $A : \Size \to U$ and $B : U$, there is an isomorphism between maps of type $A \sto TB$ and those of type $\exists i . A \, i \to B$, given by the currying isomorphism. This extends to an isomorphism between size-indexed algebra maps of type $(A, k_A) \to (TB, k_{TB})$ and algebra maps of type $(\exists i . A \, i, k_{\exists A}) \to (B, k_B)$. Thus, defining an algebra map $(\mu \F, \In) \to (B, k_B)$ is equivalent to defining a size-indexed algebra map $(\mu^\Diamond F, \In^\Diamond) \to (TB, k_{TB})$. This map is uniquely given by the initiality of $(\mu^\Diamond F, \In^\Diamond)$.  
\end{proof}

Lastly, we show that all polynomial endofunctors on small types satisfy the condition of weakly commuting with the existential quantifier.
The equivalences that we have axiomatised entail the following chain of equivalences, which corresponds precisely to the canonical map denoted in \autoref{def:weakly-commute-exist}, noting that $i \notin \text{fv}(A)$ and $i \notin \text{fv}(B(a))$ for $a : A$:
\begin{align*}
    \exists i . P_{A,B}(\Diamond X) \, i &\coloneq  \exists i \. \Sigma {a : A}  \. (B(a) \to \exists j < i . X \, j) \\
    &\simeq \Sigma {a : A} \. (\exists i . B(a) \to \exists j < i . X \, j)\\
    &\simeq \Sigma {a : A} \. (B(a) \to \exists i . \exists j < i . X \, j)\\
    &\simeq \Sigma {a : A} \. (B(a) \to \exists i . X \, i) =: P_{A,B}(\exists i . X \, i).
\end{align*}
Since all polynomial functors (both finitely- and infinitely-branching) weakly commute with the existential quantifier, their initial algebra is constructed via existential quantification over the corresponding initial size-indexed algebra. This allows us to construct W-types via quantification over size-indexed types.

\subsection{Coinductive types}
The coalgebra construction is similar to that of algebras. We define size-indexed coalgebras and construct their terminal coalgebras using the fixpoint operator. This allows us to obtain terminal coalgebras for functors satisfying an analogous property: weakly commuting with the universal quantifier. This is done by showing that the universal quantifier lifts to a functor from size-indexed coalgebras to coalgebras, which is right adjoint to the inclusion map from coalgebras to size-indexed coalgebras. 
The equivalences allow us to show that all polynomial functors weakly commute with the universal quantifier.
The details are given in \autoref{appendix:encoding}.

    \section{Realisability model}
    \label{sec:model}

\subsection{Assemblies and partial equivalence relations}

In this realisability model, based on Hyland's effective topos, \cite{hyland_effective_1982}, types are interpreted as \emph{assemblies}: sets with associated computability data in the form of a relation between a set of \emph{realisers} and elements of the set.
The intuition given by Reus \cite{REUS1999128} is that realisers can be seen as `codes' in a low-level language or a general model of computation. A code that is related to an object in the set is said to realise it, and can be seen as a representation of that object. Maps between such assemblies must be computable in this general model of computation. 

The definition of an assembly is parameterised by a general model of computation, which is given by a partial combinatory algebra (pca) \cite{feferman}, providing the set of realisers. We will restrict ourselves to the most well-known pca, the first Kleene algebra \( \mathcal{K}_1 \), although the construction should work in general.
Here, the realisers are the natural numbers, which enumerate the partial recursive functions \(\{\varphi_n\}_{n \in \bN} \).  We define a partial operation \( \cdot : \bN \times \bN \rightharpoonup \bN \) on this set of realisers by \( m \cdot n := \varphi_m(n) \), and we write \( {m \cdot n \downarrow} \) when \( \varphi_m(n) \) is defined. 
\begin{definition}
    An \emph{assembly} is a pair $(X, \Vdash_X)$, where $X$ is a set and ${\Vdash_X} \subseteq \bN \times X$ is a relation such that for every \( x \in X \) there exists some \( n \in \bN \) such that \( n \Vdash_X x \). We will call \( \Vdash_X \) the \emph{realisability} relation of \( X \). A morphism between assemblies \((X, \Vdash_X)\) and \( (Y, \Vdash_Y) \) is a function \( f : X \to  Y \) such that there exists an \( e \in \bN \) (not part of the data) satisfying:
\[
    \forall n \in \bN.\,\forall x \in X.\,(n \Vdash_X x) \implies (e \cdot n \downarrow \mbox{ and } e \cdot n \Vdash_Y f(x)).
\] 
In this case, we say that \(f\) is \emph{tracked} by \(e\). We denote the category of assemblies by \( \Asm \).
\end{definition}
Note that there is a full and faithful inclusion from sets to assemblies defined by assigning the largest realisability relation; we denote this by \( \nabla X \coloneqq (X, \bN \times X) \). 

Since morphisms between assemblies are functions that are tracked by a realiser, they must be sufficiently computable.
This ensures that the dependent product can be small, even when quantifying over large types, which allows us to model the impredicative universal quantifier.
It also allows us to enforce a notion of parametricity in the model via the inclusion from sets to assemblies. The intuition is that \( \nabla X \) carries no meaningful computational data, making the objects of \( X \) indistinguishable when defining morphisms. Since we view the axioms that we added to the theory as a kind of parametricity, this construction allows us to enforce these axioms in the model. 

Since we will interpret types as assemblies, we need a sufficiently `small' subcategory to interpret the small types. As a step towards this goal, we first consider the following subcategory of assemblies.
\begin{definition}
An assembly $(X, \Vdash_x)$ is called a \emph{modest set} if every $n \in \bN$ codes at most one $a \in X$: 
\[
    \forall n \in \bN \. (n \Vdash_X x_1, x_2) \implies(x_1 = x_2).
\]
\end{definition}
We denote the category of modest sets by $\Mod$, which is a full subcategory of $\Asm$. However, the size of the carrier set of a modest set is not bounded in any way, so we cannot equip the `set' of all modest sets with the structure of an assembly.
Instead, we note that the elements of a modest set are not that important: we only need to know when two realisers encode the same element, as this lets us represent an object by the equivalence class of its realisers.
This insight gives us a category that is equivalent to the modest sets, namely that of partial equivalence relations on the set of realisers. 
\begin{definition}
    A \emph{partial equivalence relation (PER)}  is a symmetric and transitive relation $R \subseteq \mathbb{N} \times \mathbb{N}$. We define the following notation: 
\begin{align*}
	\rm{dom}(R) &:= \{n \in \mathbb{N} \mid n\mathbin Rn\}, &
	[n]_R &:= \{m \in \mathbb{N} \mid nRm\}, &
	\mathbb{N} /R &:= \{[n]_R \mid n \in \rm{dom}(R)\}.
\end{align*} 
A morphism between PERs R and S is a function \( f : \bN / R \to \bN /S \) such that \( f([n]_R) = [e \cdot n]_S \) for some \( e \in \bN \), in which case we say that \( f \) is tracked by \( e \).
We denote this category by \( \Per \).
\end{definition}
Note that a PER R is an equivalence relation on \( \rm{dom}(R) \).
\begin{lemma}\label{lemma:per-mod-equiv}
There is an equivalence between the categories $\Per$ and $\Mod$.
\end{lemma}
\begin{proof}[Proof Sketch]
    Given a PER \( R \subseteq \bN \times \bN \) the pair $\iota (R) \coloneqq  (\bN/R, \in)$ is a modest set.
    Conversely, given a modest set \( (X, \Vdash_X) \), we can define a relation on the set of realisers as follows: \begin{align*}
            n_1\sim_{X}n_2 \iff \exists x \in X \. (n_1 \Vdash_X x) \land (n_2 \Vdash_X x).
    \end{align*}
    This induces an equivalence between PERs and modest sets.
\end{proof}

In our model, we will use assemblies to interpret large types, while \autoref{lemma:per-mod-equiv} allows us to make use of the category of PERs to interpret the small types, since every PER induces a modest set with a countable carrier. 
To describe our model, we make use of the categories with families (CwF) framework, which closely follows the syntactic structure of dependent type theory.

We recall the basic parts of which a CwF \( \mathcal{C} \) consists:  
\begin{itemize}
    \item A category \( \C \) of \emph{semantic contexts} and context morphisms;
    \item For every semantic context \( \Gamma \in \C \), a collection \( \Ty(\Gamma) \) of \emph{semantic types};
    \item For every semantic context \( \Gamma \in \C \) and semantic type \( A \in \Ty(\Gamma) \), a collection \( \Tm(\Gamma, A) \) \emph{of semantic terms};
    \item For every context morphism $f : \Delta \to \Gamma$, an operation $-[f] : \Ty(\Gamma) \to \Ty(\Delta)$ for \emph{semantic substitution} on types, and for $A \in \Ty(\Gamma)$, an operation $-[f] : \Tm(\Gamma, A) \to \Tm(\Delta, A[f])$ for semantic substitution on terms;
    \item For every semantic context $\Gamma \in \C$ and $A \in \Ty(\Gamma)$, an extended context $\Gamma.A \in \C$. 
\end{itemize}

The category of assemblies can be given the structure of a CwF, with semantic contexts given by assemblies. A semantic type \( A \in \Ty(\Gamma) \) is given by a \( \Gamma \)-indexed family of assemblies \( (A_\gamma)_{\gamma \in \Gamma} \). Given a semantic type \( A \in \Ty(\Gamma) \), a semantic term \( a \in \Tm(\Gamma, A) \) is given by a \emph{dependent morphism} \( a : (\gamma \in \Gamma) \to A_\gamma \), that is, a dependent function tracked by an \(e\in\bb N\): 
\[
    \forall n \in \bN.\,\forall \gamma \in \Gamma \. (n\Vdash_\Gamma \gamma) \implies (e\cdot n\Vdash_{A_\gamma}a(\gamma)).
\]
We will also write this as a family \( (a_\gamma)_{\gamma \in \Gamma} \). The semantic substitution operation $-[f]$ on both types and terms is given by precomposing with $f$. Finally, context extension is given by dependent pairing (as in the ussual set-theoretic model).

We will omit the term `semantic', blurring the distinction between syntax and semantics. Further details on categories with families and the associated structures in assemblies can be found in, for example, \cite{Hofmann_1997, REUS1999128}. Assemblies support the interpretation of our basic types, dependent products and sums, and the identity type. In the following subsections, we highlight the relevant constructions.

\subsection{Universe type}
We give a brief description of how the universe type is interpreted in the model, as it will provide intuition on why we interpret the type of sizes in the same way.
The universe type \( U \) in a context \( \Gamma \) is interpreted as a semantic type \( U \in \Ty(\Gamma) \), meaning a \( \Gamma \)-indexed family of assemblies. Since the universe type is a non-dependent type, its interpretation should be a constant family. As we will want to interpret the small types using PERs, we will interpret \( U \) as the constant family of an assembly where the carrier set is the set of all PERs, which we also denote by \( \Per \). To construct an assembly out of \( \Per \), we use the inclusion from sets to assemblies, giving us \( \nabla \Per \). Recall that the intuition is that \( \nabla \Per \) carries no meaningful computational data, making the PERs computationally indistinguishable when defining morphisms.
The decoder \( \El \) is interpreted by the inclusion of PERs into assemblies.

\subsection{Type of sizes}
In previous work \cite{parametric-quantifier}, the type of sizes is interpreted as the set of natural numbers. By contrast, the impredicativity of our realisability model allows us to interpret the type $\Size$ as a large type, while still having universal quantification over $\Size$ with a small type live in $\U$. Given a context $\Gamma$, we define the semantic type $\Size \in \Ty(\Gamma)$ as $\Size_\gamma = \nabla \omega_1$, where $\omega_1$ is the first uncountable ordinal: the set of all countable ordinals.
We construct an assembly from the ordinal \( \omega_1 \) by giving it the full realisability relation, just as we did for the interpretation of the universe type. This ensures that morphisms from \( \nabla \omega_1 \) must treat all sizes uniformly.

The reason for interpreting the type of sizes as $\omega_1$ is that it allows us to justify the following equivalence:
\begin{align*}
	\Pi{x : A} \. \exists i . B(x, i)&\simeq\exists i\. \Pi {x : A} \. \exists j < i . B(x, j). \tag{\(\exists\) parametric over \(\Pi\)}
\end{align*}
Intuitively, the left-to-right direction of this equivalence states the following: given a dependent function that sends a term \( a : A \) to a pair $\langle s, b \rangle : \exists i . B(a, i)$, there is a limit size $\alpha$ such that $s < \alpha$ for each $\brac{s, b}$ in the range of the dependent function. This requires limits of sequences of sizes to be a size as well. Note that, since $A$ here is a small type, it is interpreted as a partial equivalence relation over $\bN$, hence it has countably many equivalence classes. Since countable limits of countable ordinals are themselves countable, $\omega_1$ suffices for our purpose.

The ordering on sizes is interpreted as the ordering on \( \omega_1 \), interpreting \( s_1 \le s_2 \) as the terminal singleton assembly when true, and the initial empty assembly otherwise. This interpretation validates the encoding of the ordering on sizes as a proposition in the theory.

\subsection{Bounded quantification and the fixpoint operator}
Recall that the notation $i : \Size \vdash \forall j < i \. A(j)$ for bounded universal quantification is shorthand for $i : \Size \vdash \forall j \. (j < i \to A(j))$. A term of this type is interpreted by a family of dependent morphisms, where the domain is bounded by the size in context, which confirms the syntactic intuition for this type.

%

Given a term \( f : \forall i. (\forall j < i. A(j)) \to A(i) \), the fixpoint operator \( \fix \, f : \forall i . A(i) \) is interpreted as a morphism, where the underlying dependent function is defined by well-founded induction on \( \omega_1 \). The definition given by well-founded induction corresponds exactly to the \( \beta \)-rule in the theory. In order for \( \fix \, f \) to be a morphism, it must be tracked by a realiser. The following fact, which states that a fixpoint operator exists in every partial combinatory algebra, allows us to find an appropriate realiser for the fixpoint operator.
\begin{proposition}[Fixpoint operator in a pca \cite{oosten}]
    For every pca $(\cA, \cdot)$, there exists a term $\fix\! \in \cA$ such that the following hold for any \( f, a \in \cA \) we have $\fix\,f \downarrow$ and $\fix\,f \, a \simeq f \, (\fix\,f) \, a$.
\end{proposition}

\subsection{Universal and existential type}
The assembly model supports the interpretation of general impredicative products. Although we do not add impredicative products to the theory in full generality, this does allow us to interpret the (impredicative) universal quantifier $\forall i . A(i)$ where $A : U$ in the model simply as a dependent product. The interpretation of sizes as $\nabla \omega_1$ then suffices to validate the equivalences for the universal quantifier.

To interpret the existential type \( \exists i.A(i) \), we construct a quotient of the corresponding \( \Sigma \)-type that lives inside the universe. Thus, we should canonically construct a family of PERs from the semantic $\Sigma$-type, \( \Sigma(\Size, A) \), which is a family of assemblies defined in \autoref{appendix:model}. The following lemma gives us a quotient construction resulting in a modest set.
\begin{lemma}
    The inclusion functor \( I : \Mod \to \Asm \) has a left adjoint \( \M : \Asm \to \Mod \).
\end{lemma}
\begin{proof}[Proof Sketch]
    For an assembly $(A, \Vdash_A)$, we define a modest set $\M(A) = ( A {/ \!\sim_\M}, \Vdash_{\M(A)})$ where $\sim_\M$ is the least equivalence relation satisfying the following: $\exists n \in \bN \. (n \Vdash_A a_1, a_2) \implies (a_1 \sim_\M a_2)$.
The relation is defined by
\(
n \Vdash_{\M(A)} [ a ]_{\sim_\M} \Iff \exists a' \in [ a ]_{\sim_\M} (n \Vdash_A a').\)
\end{proof}
However, the universe consists of PERs rather than modest sets. Following \autoref{lemma:per-mod-equiv}, every modest set is isomorphic to a PER. Thus, we may treat \( \M(A) \) as the PER that is isomorphic to it. Given a type \( A \in \Ty(\Gamma) \), we can define its \(U\)-\emph{truncation} \( \Tr{A} \in \Tm(\Gamma, U) \) as the family of PERs \( \Tr{A}_\gamma = \M(A_\gamma) \). 
Given a small type \( A \in \Tm(\Gamma, U) \), this allows us to interpret the small type \( \exists i. A(i) \in \Tm(\Gamma, U)\) as the \(U\)-truncation of the \( \Sigma \)-type: 
\[
    \exists i. A(i) \coloneqq \Tr{\Sigma(\Size, \El A)}.
\] 

This also allows us to interpret the pairing operation and elimination of existential types. Moreover, it validates the equivalences we have axiomatised in the theory related to the existential quantifier in the model. Further details of the model are given in \autoref{appendix:model}.

    \section{Conclusion} 
    \label{sec:conclusion}
    \DeclareFontFamily{U}{matha}{\hyphenchar\font45}
\DeclareFontShape{U}{matha}{m}{n}{
    <5> <6> <7> <8> <9> <10> gen * matha
    <10.95> matha10 <12> <14.4> <17.28> <20.74> <24.88> matha12
}{}
\newcommand\bbl{\mathrel{\ooalign{\usefont{U}{matha}{m}{n}\symbol{\string"CE}}}}

We continue the development of sized types with parametric quantifiers over sizes.
In previous work \cite{licentiate, parametric-quantifier}, the type of sizes was interpreted in the semantics as the natural numbers.
This means that the following axiom for small types was restricted to the case where \(A\) is finite: \begin{align*}
    \Pi{x : A} \. \exists i . B(x, i)&\simeq\exists i\. \Pi {x : A} \. \exists j < i . B(x, j). \tag{\(\exists\) parametric over \(\Pi\)}
\end{align*}
Our main contribution is a model where the type of sizes is instead interpreted as a larger ordinal.
Indeed, for infinite \(A\), the left-to-right direction of the equivalence intuitively shows that limit ordinals exist.
This model allows us to remove the restriction that \(A\) is finite, which enables us to define all (co)inductive types, not just the finitely-branching ones.

Our model is based on realisability instead of the reflexive graphs used in previous work \cite{licentiate, parametric-quantifier}.
This semantics fits well with the intuition of parametric quantifiers: \(\exists i.\,A(i)\) behaves like the union of the \(A(i)\), while \(\forall i.\,A(i)\) behaves like the intersection of the \(A(i)\).
The model also shows that our syntax is consistent with two additional principles: an impredicative universe of sets, and uniqueness of identity proofs.
However, we do not use these properties in the syntax, which makes our construction of (co)inductive types more widely applicable and compatible with homotopy type theory.

\subparagraph*{Discussion and Future work.}

A main motivation for our work is obtaining a consistent implementation of sized types in proof assistants.
Encoding (co)inductive types using quantifiers would allow the closure ordinal \(\infty\) to be removed from \texttt{Agda}, which is a source of inconsistency of the current implementation.
However, proof assistants should satisfy desirable computational properties like decidability and canonicity, and this requires computational behaviour for our parametricity axioms.
So, to use our results for this purpose, we need to connect them with notions of internal parametricity that compute, such as those based on bridge types \cite{parametric-quantifier}, which have also been developed in combination with cubical type theory \cite{cavallo2021internal,parametricity-agda}.

In addition, we need a computational version of the $\beta$-rule for the fixpoint operator.
We have only assumed a propositional version of this rule because a na\"ive definitional version of this rule breaks normalisation due to unrestricted unfolding.
For example, it would be possible to disallow unfolding under binders, only allowing unfolding when a smaller size exists in context, or to ignore computation on sizes entirely.
However, each of these options come with drawbacks and it is not clear to us what the most principled approach is.

Because we include suprema sizes, it is not generally possible to keep inequality decidable.
In the current \(\sf{Agda}\) implementation, \(i<j\) is a judgement that is automatically checked when applying a function \(f : (i : \sf{Size}_{<j}) \to A\). However, without decidability of inequality we need to ask the user for a term of the type \(i<j\), which increases the amount of manual work.
One partial solution is to have two versions of inequality: a propositional type \(i<j\) and a definitional judgemental fragment \(i\bbl j\) (similar to the two versions of equality, \(=\) and \(\equiv\)).
Then we could add a term \(\sf{ineq}\) so that \(\sf{ineq}:i<j\) holds whenever \(i\bbl j\) (similar to the term \(\sf{refl}\) where \(\refl:a=_Aa'\) holds whenever \(a\equiv a'\)).
This would allow the user to fill in a simple term when the computer can decide the inequality, while still allowing more complicated terms when using suprema sizes.

On the semantic side, it is worth investigating whether we can generalise our realisability model.
We restricted the model to the pca of partial recursive functions, as it allowed us to interpret the type of sizes as the first uncountable ordinal.
In principle, it seems that the model could generalise to other pcas, as long as the type of sizes is interpreted as a sufficiently large ordinal: one of larger cardinality than the pca. The realizability model could also be constructed using more general notions of computation, such as the monadic combinatory algebras introduced in \cite{Cohen-et-al-2025}, to allow for an integration of sizes with effects.

Another question is whether we can model both sizes and the univalence axiom. A possibly suitable model for this is given by cubical assemblies \cite{cubical-assemblies}, which form a model of a type theory with a univalent and impredicative universe. 
It seems likely that this construction can be extended to include sizes and parametric quantifiers in a similar way as in this work, but a more careful investigation is needed to confirm this.

We would also like to investigate how realisability models can model internal parametric quantifiers and related modalities such as irrelevance in dependent type theory more generally.
The inclusion from sets to assemblies assigns a set the complete realisability relation, which intuitively makes elements computationally indistinguishable and enables parametricity for sizes.
It seems reasonable that this construction would work for a more general notion of parametric quantifier in the theory.

    \bibliography{references}
    
    \appendix
    \section{The type theory}
\label{appendix:type-theory}
The symbol $\equiv$ denotes definitional equality in the theory for both types and terms. We make use of the following forms of judgements:
\begin{align*}
    \vdash \Gamma \; \Ctxt, \qquad \Gamma \vdash A \Type, \qquad \Gamma \vdash A \equiv B \Type, \qquad \Gamma \vdash a : A, \qquad \Gamma \vdash a \equiv b : A.
\end{align*}
For brevity, we omit the standard rules stating that definitional equality is a congruence relation with respect to all type- and term formers. We may also assume weakening and substitution rules with respect to all judgements, which are both admissible \cite{Hofmann_1997}.

The rules for context formation are given as follows:
    \begin{center}
        \begin{bprooftree}
            \AxiomC{$\phantom{\Gamma, x : A \vdash M : B}$}
            \UnaryInfC{$\vdash \diamond \; \Ctxt$}
        \end{bprooftree}\quad
        \begin{bprooftree}
            \AxiomC{$\vdash \Gamma \; \Ctxt$}
            \AxiomC{$\Gamma \vdash A \Type$}
            \BinaryInfC{$\vdash \Gamma, x : A \; \Ctxt$}
        \end{bprooftree}
    \end{center}
    \begin{center}
        \begin{bprooftree}
            \AxiomC{$\vdash \Gamma, x : A, \Delta \; \Ctxt$}
            \UnaryInfC{$\Gamma, x:A , \Delta \vdash x:A$}
        \end{bprooftree}\quad
        \begin{bprooftree}
            \AxiomC{$\Gamma, x : A\vdash b(x):B(x)$}
            \AxiomC{$\Gamma\vdash a:A$}
            \BinaryInfC{$\Gamma\vdash b(a):B(a)$}
        \end{bprooftree}
    \end{center}
\vspace{10pt}

The rules for the empty-, unit-, and boolean type are as follows, where their elimination and computation principles are left implicit.
\begin{center}
    \begin{bprooftree}
        \AxiomC{$\phantom{\Gamma \vdash M : \Pi x : A.\,B}$}
        \UnaryInfC{$\Gamma \vdash \bot \Type$}
    \end{bprooftree}\quad
    \begin{bprooftree}
        \AxiomC{$\phantom{\Gamma \vdash M : \Pi x : A.\,B}$}
        \UnaryInfC{$\Gamma \vdash \top \Type$}
    \end{bprooftree}\quad
    \begin{bprooftree}
        \AxiomC{$\phantom{\Gamma \vdash M : \Pi x : A.\,B}$}
        \UnaryInfC{$\Gamma \vdash \mathsf{Bool} \Type$}
    \end{bprooftree}\quad
\end{center}
\begin{center}
    \begin{bprooftree}
        \AxiomC{\phantom{$\Gamma, x : A \vdash b : B$}}
        \UnaryInfC{$\Gamma \vdash \star : \top$}
    \end{bprooftree}\quad
    \begin{bprooftree}
        \AxiomC{\phantom{$\Gamma, x : A \vdash b : B$}}
        \UnaryInfC{$\Gamma \vdash \mathsf{tt} : \mathsf{Bool}$}
    \end{bprooftree}\quad
    \begin{bprooftree}
        \AxiomC{\phantom{$\Gamma, x : A \vdash b : B$}}
        \UnaryInfC{$\Gamma \vdash \mathsf{ff} : \mathsf{Bool}$}
    \end{bprooftree}\quad
\end{center}
\vspace{10pt}

The rules for $\Pi$-types include definitional $\beta$- and $\eta$-rules.
\begin{center}
    \begin{bprooftree}
        \AxiomC{$\Gamma \vdash A \Type$}
        \AxiomC{$\Gamma, x : A \vdash B(x) \Type$}
        \BinaryInfC{$\Gamma \vdash \Pi x : A.\,B(x) \Type$}
    \end{bprooftree}\quad
\end{center}
\begin{center}
    \begin{bprooftree}
        \AxiomC{$\Gamma, x : A \vdash b(x) : B(x)$}
        \UnaryInfC{$\Gamma \vdash \lambda x \: A \. b(x): \Pi x : A.\,B(x)$}
    \end{bprooftree}\quad
    \begin{bprooftree}
        \AxiomC{$\Gamma \vdash f : \Pi x : A.\,B(x)$}
        \AxiomC{$\Gamma \vdash a : A$}
        \BinaryInfC{$\Gamma \vdash f \,a : B(a)$}
    \end{bprooftree}
\end{center}
\begin{center}
    \begin{bprooftree}
        \AxiomC{$\Gamma, x : A \vdash b(x) : B(x)$}
        \AxiomC{$\Gamma \vdash a : A$}
        \BinaryInfC{$\Gamma \vdash (\lambda x \: A \. b(x)) \, a \equiv b(a): B(a)$}
    \end{bprooftree}\quad
    \begin{bprooftree}
        \AxiomC{$\Gamma \vdash f : \Pi x : A.\,B(x)$}
        \UnaryInfC{$\Gamma \vdash (\lambda x \: A \. f \, x) \equiv f : \Pi x : A.\,B(x)$}
    \end{bprooftree}
\end{center}
\vspace{10pt}

Similarly to $\Pi$-types, $\Sigma$-types also come with $\beta$- and $\eta$-rules. The $\eta$-rule for $\Sigma$-types is defined via its destructors, where $\pi_1$ and $\pi_2$ are the projections derived from the elimination principle. 
\begin{center}
    \begin{bprooftree}
        \AxiomC{$\Gamma \vdash A \Type$}
        \AxiomC{$\Gamma, x : A \vdash B(x) \Type$}
        \BinaryInfC{$\Gamma \vdash \Sigma x:A.\,B(x) \Type$}
    \end{bprooftree}\quad
    \begin{bprooftree}
        \AxiomC{$\Gamma \vdash a : A$}
        \AxiomC{$\Gamma \vdash b : B(a)$}
        \BinaryInfC{$\Gamma \vdash \< a, b \> : \Sigma x:A.\,B(x)$}
    \end{bprooftree}
    \end{center}
    \begin{center}
        \begin{bprooftree}
            \AxiomC{$\Gamma, z : \Sigma x:A.\,B(x) \vdash P(z)\Type$}
            \noLine
            \UnaryInfC{$\Gamma, x : A, y : B \vdash p(x,y) : P(\< x, y \>)$}
            \UnaryInfC{$\Gamma, z : \Sigma x:A.\,B(x) \vdash \ind_\Sigma(z, p(x,y)) : P(z)$}
        \end{bprooftree}
    \end{center}
    \begin{center}
        \begin{bprooftree}
            \AxiomC{$\Gamma, z : \Sigma x:A.\,B(x) \vdash P(z)\Type$}
            \noLine
            \UnaryInfC{$\Gamma, x : A, y : B \vdash p : P(\< x, y \>)$}
            \UnaryInfC{$\Gamma, x : A, y : B \vdash \ind_\Sigma(\< x, y \>, p) \equiv p(x,y) : P(\< x, y \>)$}
        \end{bprooftree}
    \end{center}
    \begin{center}
        \begin{bprooftree}
            \AxiomC{$\Gamma \vdash c : \Sigma x:A.\,B(x)$}
            \UnaryInfC{$\Gamma \vdash c \equiv \< \pi_1(c), \pi_2(c) \> : \Sigma x:A.\,B(x)$}
        \end{bprooftree}
    \end{center}
\vspace{10pt}

Propositional identity types are defined with the standard induction principle.
    \begin{center}
        \begin{bprooftree}
            \AxiomC{$\Gamma \vdash A \Type$}
            \AxiomC{$\Gamma \vdash a_1, a_2 : A$}
            \BinaryInfC{$\Gamma \vdash a_1 =_A a_2 \Type$}
        \end{bprooftree}\quad
        \begin{bprooftree}
            \AxiomC{$\Gamma \vdash A \Type$}
            \AxiomC{$\Gamma \vdash a : A$}
            \BinaryInfC{$\Gamma \vdash \refl \, a : a =_A a$}
        \end{bprooftree}\quad
    \end{center}
    \begin{center}
        \begin{bprooftree}
            \AxiomC{$\Gamma , x : A, y: A, z : x =_A y \vdash P(x, y, z) \Type$}
            \noLine
            \UnaryInfC{$\Gamma , x : A \vdash p(x) : P(x, x, \refl \, x)$}
            \UnaryInfC{$\Gamma, x:A, y:A, z: x=_A y \vdash \ind_=(x,y, z, p) : P(x, y, z)$}
        \end{bprooftree}
    \end{center}
    \begin{center}
        \begin{bprooftree}
            \AxiomC{$\Gamma , x : A, y : A, z : x =_A y \vdash P(x, y, z) \Type$}
            \noLine
            \UnaryInfC{$\Gamma , x : A \vdash p(x) : P(x, x, \refl \, x)$}
            \UnaryInfC{$\Gamma, x:A \vdash \ind_=(x, x, \refl \, x, p) \equiv p(x) : P(x, x, \refl \, x)$}
        \end{bprooftree}
    \end{center}
\vspace{10pt}

The universe type is closed under the base types, propositional equality, and the type constructors. The definitional type equalities for the small types ensure that they lift to the proper type.
    \begin{center}
        \begin{bprooftree}
        \AxiomC{\phantom{$\Gamma \vdash u : \U$}}
        \UnaryInfC{$\Gamma \vdash \U \Type$}
        \end{bprooftree}\quad
        \begin{bprooftree}
        \AxiomC{$\Gamma \vdash A : \U$}
        \UnaryInfC{$\Gamma \vdash \El A \, \Type$}
        \end{bprooftree}
    \end{center}
    \begin{center}
        \begin{bprooftree}
            \AxiomC{$\phantom{\Gamma \vdash}$}
            \UnaryInfC{$\Gamma \vdash \bot^\U : \U$}
        \end{bprooftree}\quad
        \begin{bprooftree}
            \AxiomC{$\phantom{\Gamma \vdash}$}
            \UnaryInfC{$\Gamma \vdash \top^\U : \U$}
        \end{bprooftree}\quad
        \begin{bprooftree}
            \AxiomC{$\phantom{\Gamma \vdash}$}
            \UnaryInfC{$\Gamma \vdash \mathsf{Bool}^\U : \U$}
        \end{bprooftree}
    \end{center}
    \begin{center}
        \begin{bprooftree}
            \AxiomC{$\Gamma \vdash A : U$}
            \AxiomC{$\Gamma, x : A \vdash B(x) : \U$}
            \BinaryInfC{$\Gamma \vdash \Pi^U x:A.\, B(x): \U$}
        \end{bprooftree}
    \end{center}
    \begin{center}
        \begin{bprooftree}
            \AxiomC{$\Gamma \vdash A : \U$}
            \AxiomC{$\Gamma, x : \El A \vdash B(x) : \U$}
            \BinaryInfC{$\Gamma \vdash \Sigma^\U x:A.\, B(x) : \U$}
        \end{bprooftree}
    \end{center}
    \begin{center}
        \begin{bprooftree}
            \AxiomC{$\Gamma \vdash A : \U$}
            \AxiomC{$\Gamma\vdash a : \El A$}
            \AxiomC{$\Gamma\vdash a' : \El A$}
            \TrinaryInfC{$\Gamma \vdash a=^U_Aa' : \U$}
        \end{bprooftree}
    \end{center}
    \begin{center}
    \begin{align*}
        \El (\bot^\U) &\equiv \bot & \El (\Sigma^U x:A.\, B(x)) &\equiv \Sigma_{x:\El A} (\El B(x)) \\
        \El (\top^\U) &\equiv \top & \El (\Pi^U x:A.\, B(x)) &\equiv \Pi_{x:\El A} (\El B(x)) \\
        \El (\mathsf{Bool}^\U) &\equiv \mathsf{Bool} & \El(a=^U_Aa')&\equiv(a=_{\El A}a')
    \end{align*}
    \end{center}
\vspace{10pt}

The type of sizes is defined with constructors for zero and successor. The ordering on sizes is defined as a mere proposition, with constructors for reflexivity and transitivity.
\begin{center}
    \begin{bprooftree}
        \AxiomC{\phantom{$\Gamma \vdash s_1, s_2 : \Size$}}
        \UnaryInfC{$\Gamma \vdash \Size \Type$}
    \end{bprooftree}\quad
    \begin{bprooftree}
        \AxiomC{$\Gamma \vdash  i, j : \Size$}
        \UnaryInfC{$\Gamma \vdash \; i \leq j : U$}
    \end{bprooftree}\quad
    \begin{bprooftree}
        \AxiomC{\phantom{$\Gamma \vdash s_1, s_2 : \Size$}}
        \UnaryInfC{$\Gamma \vdash 0 : \Size$}
    \end{bprooftree}\quad
    \begin{bprooftree}
        \AxiomC{$\Gamma \vdash i  : \Size$}
        \UnaryInfC{$\Gamma \vdash {\uparrow \! i} : \Size$}
    \end{bprooftree}
\end{center}

\begin{center}
    \begin{bprooftree}
        \AxiomC{$\Gamma\vdash i:\Size$}
        \AxiomC{$\Gamma\vdash j:\Size$}
        \BinaryInfC{$\Gamma\vdash i\le j\Type$}
    \end{bprooftree}
    \begin{bprooftree}
        \AxiomC{$\Gamma \vdash i : \Size$}
        \UnaryInfC{$\Gamma \vdash{\sf{le}_0}\,i:\El(0 \le i)$}
    \end{bprooftree}\quad
    \begin{bprooftree}
        \AxiomC{$\Gamma \vdash i : \Size$}
        \UnaryInfC{$\Gamma \vdash{\sf{le}_{\sf{suc}}}\,i:\El(i \le {\uparrow}i) $}
    \end{bprooftree} \\[2ex]
    \begin{bprooftree}
        \AxiomC{$\Gamma \vdash i : \Size$}
        \UnaryInfC{$\Gamma \vdash\sf{le}_{\sf{refl}}\,i:\El(i \le i)$}
    \end{bprooftree}\quad
    \begin{bprooftree}
        \AxiomC{$\Gamma \vdash p:i\le j$}
        \AxiomC{$\Gamma \vdash q:j\le k$}
        \BinaryInfC{$\Gamma \vdash\sf{le}_{\sf{trans}}\,p\,q:\El(i \le k)$}
    \end{bprooftree} \\[2ex]
    \begin{bprooftree}
        \AxiomC{$\Gamma \vdash p:i\le j$}
        \AxiomC{$\Gamma \vdash q:i\le j$}
        \BinaryInfC{$\Gamma \vdash\sf{le}_{\sf{eq}}\,p\,q :\El(p=q)$}
    \end{bprooftree}
\end{center}
\vspace{10pt}

The type of the existential quantifier over sizes is defined as follows. Note that the elimination principle is limited to small types.
    \begin{center}
        \begin{bprooftree}
            \AxiomC{$\Gamma, i : \Size \vdash A(i) : \U$}
            \UnaryInfC{$\Gamma \vdash \exists i . A(i) : \U $}
        \end{bprooftree}
    \end{center}
    \begin{center}
        \begin{bprooftree}
            \AxiomC{$\Gamma, i : \Size \vdash A(i) : \U$}
            \AxiomC{$\Gamma \vdash s : \Size$}
            \AxiomC{$\Gamma \vdash a : \El A(i)$}
            \TrinaryInfC{$\Gamma \vdash \langle s, a \rangle^\exists : \El (\exists i . A(i))$}
        \end{bprooftree}
    \end{center}
    \begin{center}
        \begin{bprooftree}
            \AxiomC{$\Gamma, z : \El (\exists i. A(i)) \vdash P(z) : \U$}
            \noLine
            \UnaryInfC{$\Gamma, i : \Size, x : \El A \vdash p(i,x) : \El P(\langle i, x \rangle^\exists)$}
            \UnaryInfC{$\Gamma, z : \El(\exists i.A(i))\vdash \ind_\exists(z,p(i,x)) : \El P(z)$}
        \end{bprooftree}
    \end{center}
    \begin{center}
        \begin{bprooftree}
            \AxiomC{$\Gamma, z : \El (\exists i. A(i)) \vdash P(z) : \U$}
            \noLine
            \UnaryInfC{$\Gamma, i : \Size, x : \El A \vdash p : \El P(\langle i, x \rangle^\exists)$}
            \noLine
            \UnaryInfC{$\Gamma\vdash s:\Size$\qquad$\Gamma\vdash a:\El A(s)$}
            \UnaryInfC{$\Gamma\vdash \ind_\exists(\langle s, a \rangle^\exists, p(i,x)) \equiv p(s,a) : \El P(\langle s,a \rangle^\exists)$}
        \end{bprooftree}
    \end{center}
\vspace{10pt}

The universal quantifier is defined analogously to the $\Pi$-type, only as an impredicative quantifier.
\begin{center}
    \begin{bprooftree}
        \AxiomC{$\Gamma, i : \Size \vdash A(i) :U$}
        \UnaryInfC{$\Gamma \vdash \forall i . A(i) :U$}
    \end{bprooftree}\quad
\end{center}
\begin{center}
    \begin{bprooftree}
        \AxiomC{$\Gamma, i : \Size \vdash a(i) : A(i)$}
        \UnaryInfC{$\Gamma \vdash \lambda^\forall i : \Size \. a(i): \El (\forall i . A(i)) $}
    \end{bprooftree}\quad
    \begin{bprooftree}
        \AxiomC{$\Gamma \vdash f : \El(\forall i . A(i))$}
        \AxiomC{$\Gamma \vdash s : \Size$}
        \BinaryInfC{$\Gamma \vdash f \, s : \El A(s)$}
    \end{bprooftree}
\end{center}
\begin{center}
    \begin{bprooftree}
        \AxiomC{$\Gamma, i : \Size \vdash a(i) : \El A(i)$}
        \AxiomC{$\Gamma \vdash s : \Size$}
        \BinaryInfC{$\Gamma \vdash (\lambda^\forall i \: \Size \. a(i)) \, s \equiv a(s): \El A(s)$}
    \end{bprooftree}\quad
    \begin{bprooftree}
        \AxiomC{$\Gamma \vdash f : \El(\forall i . A(i))$}
        \UnaryInfC{$\Gamma \vdash (\lambda^\forall i \: \Size \. f \, i) \equiv f : \El (\forall i . A(i))$}
    \end{bprooftree}
\end{center}
\vspace{10pt}

Finally, the fixpoint operator on sizes, corresponding to well-founded induction, is defined as follows. Note that the associated $\beta$-rule is propositional.
    \begin{center}
        \begin{bprooftree}
            \AxiomC{$\Gamma, i : \Size \vdash A(i) \Type$}
            \AxiomC{$\Gamma \vdash f : \forall i . ((\forall j < i. A(j)) \to A(i))$}
            \BinaryInfC{$\Gamma \vdash \fix f : \forall i. A(i)$}
        \end{bprooftree}
    \end{center}
    \begin{center}
    \begin{bprooftree}
            \AxiomC{$f : \forall i . ((\forall j < i. A(j)) \to A(i))$}
            \UnaryInfC{$\Gamma \vdash \fix^\beta \, f : \forall i \. \fix f \, i = f \, i \, (\lambda j < i\. \fix f  \, j)$}
        \end{bprooftree}
    \end{center}

Function extensionality is encoded as in \cite{hottbook}: given $f, g : \Pi x : A.\,B(x)$, there is a map \[
    (f = g) \to \Pi x : A.\, (f(x) = g(x)),
\] and the function extensionality axiom states that this map is an equivalence.
Similarly, we note that for small types \(A\), \(B(i,x)\), and \(C(i)\) there are canonical maps, for which the parametricity axioms state that each such map is an equivalence: 
\begin{align*}
    \Pi x : A.\,\exists i . B(x, i)&\leftarrow\exists i \. \Pi x : A. \, \exists j < i .B(x, j), \\
    \Sigma x:A.\,\forall i . B(x, i) &\to \forall i \. \Sigma x : A.\,B(i, x), \\
    \exists i \.  C(i)&\leftarrow \exists i \. \exists j < i \.C(j), \\
    \forall i \.  C(i) &\to \forall i \. \forall j < i \.C(j).
\end{align*} 

\newpage

\section{Encoding}
\label{appendix:encoding}
\subsection{Inductive types}

\begin{proof}[Proof of \autoref{lemma:inductive-sized-encoding}]
    The propositional \( \beta \)-rule for the fixpoint operator gives us the following equality:
    \begin{align*}
        \mu^\Diamond F \, i &\coloneqq \fix (\lambda i \. \lambda X\. F(\Diamond_i X)) \, i \\
        &= (\lambda i \. \lambda X\. F(\Diamond_i X)) \, i \, (\lambda r < s \. \fix (\dots) \, r) \\
        &\equiv F(\Diamond_s \, \fix (\lambda i \. \lambda X\. F(\Diamond_i X)) ) \\
        & =: F [\Diamond (\mu^\Diamond F)] \, i
    .\end{align*}
    The structure map \( \In^\Diamond : F[\Diamond (\mu^\Diamond F)] \sto \mu^\Diamond F \) is then given by the pointwise transport operation on the path given by $\fix^\beta f \, s$. This map is invertible, and we denote its inverse by $\out^\Diamond$. Given an algebra \( (A, k_A) \), the term \( \fold^\Diamond \, A \, k_A : \mu^\Diamond F \sto A\) is also defined using the fixpoint operator as follows:
    \begin{align*}
        \fold^\Diamond \, A \, k_A  &\coloneqq \fix \, (\lambda i \. \lambda \fold^{\Diamond'} \. \lambda m \. (k_A \circ F[\Diamond \fold^{\Diamond'}] \circ \mathsf{out}^\Diamond) \, i \, m).
    \end{align*}
    Thus, for $i : \Size$ and $m : \mu^\Diamond F \, i$, by $\fix^\beta$ we have
    \[
    (\fold^\Diamond \, A \, k_A) \, i \, m = (k_A \circ F[\Diamond \fold^{\Diamond}] \circ \mathsf{out}^\Diamond) \, i \, m.
    \]
    We obtain a path $\fold^\Diamond \, A \, k_A \circ \In^\Diamond = k_A \circ F[\Diamond (\fold^\Diamond \, A \, k_A)]$ using $\fix^\beta f$ and the canonical proof that $\out^\Diamond$ is the inverse of $\In^\Diamond$. 
    To show that the following is contractible: \[\Sigma{(h : \mu^\Diamond F \to A)} \. (h \circ \In^\Diamond = k_A \circ F[\Diamond h])\] is equivalent to show that the following is contractible: \[\Sigma{(h : \mu^\Diamond F \to A)} \. (h = k_A \circ F[\Diamond h] \circ \out^\Diamond).\] 
    By function extensionality, this is then equivalent to showing that the following type is contractible: \[\Sigma{(h : \mu F \to A)} \, \forall i \. \Pi{(m : \mu^\Diamond F \, i)} \. (h \, i \, m = (k_A \circ Fh \circ \out^\Diamond) \, i \, m).\]
    Since $\fold^\Diamond \, A \, k_A$ is defined via the fixpoint operator, this follows from \autoref{lemma:unique-fix}.
\end{proof}

\subsection{Coinductive types}
\begin{definition}[Coalgebra]
    Given an endofunctor \( F : U \to U \), an \emph{F-coalgebra} is a pair \( (C, k_C) \), consisting of a type \( C : U \) and a structure map \( k_C : C \to FC \). A map of coalgebras \( (h, s_h) : (C, k_C) \to (D, k_D) \) is a pair of a map \( h : C \to D \) and a path \( s_h : k_D \circ h = Fh \circ k_C\). We write $\FCoalg$ for the type of $F$-algebras.

A \emph{terminal $F$-coalgebra} is an coalgebra \( (\nu F, \out) \), such that for any coalgebra \( (C, k_C) \), the following type is contractible: 
\[\Sigma{(h : C \to \nu F)} \, (\out \circ h = Fh \circ k_C).\]
\end{definition}
We define \emph{size-indexed \( F \)-coalgebras} as coalgebras for \( F[\Box -] \). Explicitly, a size-indexed \( F \)-coalgebra is a pair \( (C, k_C) \) where \( C : \Size \to U \) and \( k_C : C \sto F[\Box C] \).

To encode size-indexed coinductive types, we construct terminal coalgebras for size-indexed functors. This is done in the following lemma, where we make use of the fixpoint operator:
\begin{lemma}\label{lemma:coinductive-sized-encoding}
    Let \( F \) be a functor on small types. The size-indexed functor \( F[\Box -] \) has a terminal coalgebra \( (\nu^\Box F, \out^\Box) \), where the first component is given by:
\[
    \nu^\Box F \, i \coloneqq  \fix (\lambda i. \lambda  X. F(\Box_i X)) \, i.
\] 
\end{lemma}
\begin{proof}[Proof Sketch]
    We note that the propositional \( \beta \)-rule for the fixpoint operator gives us the following equality:
    \begin{align*}
        \nu^\Box F \, i &\coloneqq \fix (\lambda i \. \lambda X\. F(\Box_i X)) \, i \\
        &= (\lambda i \. \lambda X\. F(\Box_i X)) \, i \, (\lambda r < s \. \fix (\dots) \, r) \\
        &\equiv F(\Box_i \, \fix (\lambda s \. \lambda X\. F(\Box_s X)) ) \\
        & =: F [\Box (\nu^\Box F)] \, i
    .\end{align*}
    The structure map \( \out : \nu^\Box F \sto F[\Box (\nu^\Box F)] \) is then given by the pointwise transport operation. Given a coalgebra \( (C, k_C) \), the term \( \unfold^\Box \, C \, k_C : C \sto \nu^\Box F \) is also defined using the fixpoint operator in the same way as in the algebra case, and can be shown to be an coalgebra morphism. The fact that it is a terminal coalgebra then follows from the contractibility of \( \fix \!\) described in \autoref{lemma:unique-fix}, together with function extensionality.
\end{proof}

As for algebras, we can lift a coalgebra to a size-indexed coalgebra. We reuse the notation $T$ from algebras.
\begin{lemma}
    Every \( F \)-coalgebra \( (C, k_C) \) gives rise to an \( F[\Box{-}] \)-coalgebra \( (TC, k_{TC}) \), where the first component is defined by 
    \[
        TC \coloneqq \lambda i. C : \Size \to U.
    \] 
This gives rise to a functorial map \( T : \FCoalg \to \sFCoalg \). 
\end{lemma}
\begin{proof}[Proof Sketch]
    We require a structure map \( k_{TC} : TC \sto F[\Box TC]\). Noting that \( i \notin \text{fv}(C) \), we can define the following auxiliary term:
    \begin{align*}
        \mathsf{insert}& : \Pi{C:\U} \. TC \sto \Box TC, \\
        \mathsf{insert}&_C \, i \, c \coloneq  \lambda j < i \. c
    \end{align*}
    This allows us to define $k_{TC} \, i \coloneq  F(\mathsf{insert}_C \, i) \circ k_C$, which is well-typed.
    Given a coalgebra morphism $f : (C, k_C) \to (D, k_D)$, we define $Tf : ({TC}, k_{TC}) \to ({TD}, k_{TD})$ by
    \begin{align*}
        Tf &: TC \sto TD, \\
        Tf & \coloneq  \lambda i . f.
    \end{align*}

    By unfolding the definitions, it is clear that $\mathsf{insert}$ is natural in $C$. As in the algebra case, this is sufficient to show that $Tf$ satisfies the required equality for a coalgebra morphism. Thus, $T$ as defined here is a functorial map from $F\text-\sf{CoAlg}$ to $F[\Box{-}]\text-\sf{CoAlg}$.
\end{proof}

The analogous property that we require a functor to satisfy in order to construct a terminal coalgebra is given as follows:

\begin{definition}
\label{weakly-commute-forall}
    Let $\F$ be a functor on the category of small types. We say that $\F$ \emph{weakly commutes} with $\forall$ if the following canonical map
    \begin{align*}
    \mathsf{can}_C &: \F (\forall i .C \, i) \to (\forall i \. F [\Box C] \, i), \\
    \mathsf{can}_C &\, c' \, i \coloneq  F \, (\psi_C \, i) \, c' .
    \end{align*}
    is an equivalence for each $C : \Size \to \U$, where we define $\psi_C$ as:
    \begin{align*}
        \psi_C &: \forall i \. (\forall j \. C \, j \to \forall j < i \. C \, j), \\
        \psi_C & \, i \, c' \, j \coloneq  c' j.
    \end{align*}
\end{definition}

For functors that satisfy this property, we can define a functor $\forall : F[\Box{-}]\text-\sf{CoAlg} \to F\text-\sf{CoAlg}$ extending the action of the universal quantifier on size-indexed types as follows.
\begin{proposition}
    Given a functor $F$ on the category of small types that weakly commutes with $\forall$, every $F[\Box-]$-coalgebra $(C, k_C)$ gives rise to a type $\forall i . A \, i$ which is the first component of an $F$-coalgebra. Moreover, this operation is functorial.
\end{proposition}
\begin{proof}
    First, we note that the universal quantifier works on maps too, meaning that for any map $f : C \sto D$, we define a map betewen universally quantified types as:
    \begin{align*}
        &\forall i . C \, i \to \forall i . D \, i, \\
        &\forall f \coloneq  \lambda i . f \, i \, (g \, i).
    \end{align*}

    Now let $\mathsf{can}_C^{-1} : \forall i \. F [\Box C] \, i \to \F (\forall i .C \, i) $ be the isomorphism obtained from $F$ weakly commuting with $\forall$.
    For an $F[\Box-]$-coalgebra $(C, k_C)$, where $k_C : C \sto F[\Box C]$, we obtain the $F$-coalgebra $(\forall i . \, C \, i \, , k_{\forall C})$ where we define
    \begin{align*}
        k_{\forall C}& : \forall i . \, C \, i \to F(\forall i . \, C \, i), \\
        k_{\forall C}& \coloneq  \mathsf{can}_C^{-1} \circ \forall  k_C.
    \end{align*}

    Given an coalgebra morphism $f: (C, k_C) \to (D, k_D)$, we obtain the map
    $\forall f : \forall i . C \, i \to \forall i . D \, i$. The fact that this constitutes a coalgebra morphism follows from the fact that $\mathsf{can}_C$ is natural in $C$.
\end{proof}

Dually to the case for initial algebras, we are able to prove that for functors that weakly commute with $\forall$, their terminal coalgebra is obtained by universal quantification over the terminal size-indexed coalgebra.
\begin{theorem}
    Let $\F$ be a functor on the category of small types that weakly commutes with the universal quantifier. Then $\F$ has a terminal coalgebra with first component:
    \[
    \nu  \F \coloneq  \forall i . (\nu^\Box F) \, i.
    \]
\end{theorem}
\begin{proof}
This argument is also very similar to the algebra case.
The categorical intuition is that this follows from the fact that the functor $\forall : \sFCoalg \to \FCoalg$ of the previous proposition is the right adjoint of $T : \FCoalg \to \sFCoalg$.
Since right adjoints preserve terminal objects, we obtain the fact that $\forall i. (\nu^\Box F) \, i$ is the terminal coalgebra of $F$, with the structure map $\out \coloneq  \mathsf{can}^{-1}_{\nu^\Box F} \circ \forall \, \In^\Box$, which is also an equivalence.

Working this out in more detail: given a type $C : U$ and a size-indexed type $D : \Size \to U$, there is an isomorphism between maps of type $TC \sto D$ and those of type $C \to \forall i . D \, i $, which is given by swapping arguments. This extends to an isomorphism between size-indexed coalgebra maps of type $(TC, k_{TC}) \to (D, k_{D})$ and coalgebra maps of type $(C, k_C) \to (\forall i . D \, i, k_{\forall D})$. Thus, defining a coalgebra map $(C, k_C) \to (\nu \F, \out) $ is equivalent to defining an size-indexed coalgebra map $(C, k_{TC}) \to (\nu^\Box F, \out^\Box)$. Such a coalgebra map uniquely exists by the fact that $(\nu^\Box F, \out^\Box)$ is the terminal size-indexed coalgebra.  
\end{proof}

Finally, we can show that all polynomial functors satisfy the requirement of weakly commuting with the universal quantifier. As before, the equivalences entail the following chain of equivalences, which corresponds precisely to the canonical map denoted in \autoref{weakly-commute-forall}, noting that $i \notin \text{fv}(A)$ and $i \notin \text{fv}(B(a))$ for $a : A$:
\begin{align*}
    P_{A,B} (\forall i .X \, i) &\coloneq  \Sigma {a : A} \. (B (a) \to \forall i .X \, i) \\
    &\simeq \Sigma {a : A} \. (B (a) \to \forall i . \forall j < i . X \, j) \\
    &\simeq \Sigma {a : A} \. \forall i \. B(a) \to \forall j < i . X \, j \\
    &\simeq \forall i \. \Sigma {a : A} \. B(a) \to \forall j < i . X \, j
    =: (\forall i \. P_{A,B} [\Box X] \, i).
\end{align*}
Thus, for all polynomial functors, the terminal coalgebra is constructed via universal quantification.

\section{Model of the theory}
\label{appendix:model}
In this section we give more details on how the category of assemblies fits in the CwF framework, and how our additional type formers are interpreted in the model. Recall that we work with the pca of partial recursive functions. We also assume a computable bijective pairing function $\langle \_\, , \_ \rangle : \bN^2 \to \bN$ is given, together with projections $\pr_1, \pr_2 : \bN \to \bN$. An enumeration of partial recursive functions and an appropriate pairing function can be found in many treatments of recursion theory. Another important feature of a pca is that it has a notion of abstraction similar to $\lambda$-abstraction in the lambda calculus. For a variable $x$ and a term $t$, there exists another term $\Lambda x . t$ defined by recursion on terms, which behaves exactly as $\lambda$-abstraction. See, for example, \cite{oosten} for more details on partial combinatory algebra.

\subsection{Assemblies as a CwF}
 The CwF framework can be defined in different ways, we make use of the definitions from Hofmann \cite{Hofmann_1997}, in which the construction of a CwF begins with the following structure:
\begin{itemize}
	\item A category $\C$ of semantic contexts and context morphisms;
	\item For each $\Gamma \in \C$, a collection $\Ty(\Gamma)$ of semantic types in context $\Gamma$.
	\item For each $\Gamma \in \C$ and $A \in \Ty(\Gamma)$, a collection $\Tm(\Gamma, A)$ of semantic terms of type $A$ in context $\Gamma$.
\end{itemize}
We will often leave out ``semantic'' where this does not cause confusion. The category $\Asm$ is the category of contexts. Given an assembly $\Gamma$, a type $A \in \Ty(\Gamma)$ will be a family of assemblies over $\Gamma$ We will often write such a family as $A = (A_\gamma)_{\gamma \in \Gamma}$. Finally, a term $a \in \Tm(\Gamma, A)$ is a dependent morphism $a : (\gamma \in \Gamma) \to A_\gamma$, often written $a = (a_\gamma)_{\gamma \in \Gamma}$, which means that $a$ is tracked by some $e \in \bN$, i.e.
\[
\forall n \in \bN . \forall \gamma \in \Gamma \. n \Vdash_\Gamma \gamma \implies (e \cdot n) \Vdash_{A_\gamma} a_\gamma.
\]

A CwF requires a notion of \emph{semantic substitution}, given by an operation on types and one on terms. Given a context morphism $f :  \Delta \to \Gamma$, there should be an operation $-[f] : \Ty(\Gamma) \to \Ty(\Delta)$, and for each $A \in \Ty(\Gamma)$ an operation $-[f] : \Tm(\Gamma, A) \to \Tm(\Delta, A[f])$. Moreover, these operations should be compatible with identities and composition \cite{Hofmann_1997}. 
For assemblies, both operations are defined by precomposition with $f$. Given an assembly morphism $f :  \Delta \to \Gamma$ and a type $A \in \Ty(\Delta)$, we define $A[f]_\delta = A_{f(\delta)}$. Similarly, if $a \in \Tm(\Gamma, A)$, we define $(a[f])_\delta = a_{f(\delta)}$.

A CwF should be able to interpret the \emph{context formation} rules of the theory. The empty context $\diamond$ is modelled by the terminal object in $\C$. In the category of assemblies, the terminal object is given by $\1 = (\{*\}, \nabla \{*\})$. To interpret context extension, we require for each $\Gamma \in \C$ and $A \in \Ty(\Gamma)$ an extended context $\Gamma.A \in \C$, together with context morphisms $\p_A : \Gamma.A \to \Gamma$ and $\v_A \in \Tm(\Gamma.A, A[\p_A])$, corresponding to first and second projections respectively. Context extension is defined as follows for assemblies:
\begin{definition}[Context extension]
	Given an assembly $\Gamma$ and $A \in \Ty(\Gamma)$, we define the context extension $\Gamma.A$ as:
	\begin{align*}
		& \Gamma.A = \{ (\gamma, a) \mid \gamma \in \Gamma \land a \in A_\gamma\}, \\
		& \langle n_1, n_2 \rangle \Vdash_{\Gamma.A} (\gamma, a) \Iff n_1\Vdash_\Gamma \gamma \and n_2 \Vdash_{A_\gamma} a.
	\end{align*}
	The morphisms $\p_A : \sigma(\Gamma, A) \to \Gamma$ and $\v_A \in \Tm(\Gamma.A, A[\p_A])$ are given by $(\gamma, a) \mapsto \gamma$ and $(\gamma, a) \mapsto a$ respectively, which are assembly morphisms tracked by the first and second projections on the realisers. It is moreover clear that in this model, $A[\p_A]_{(\gamma, a)} = A_\gamma$. 
\end{definition}

Finally, a CwF requires an operation to extend a context morphism by a term. That is, given a context morphism $f : \Delta \to \Gamma$, a type $A \in \Ty(\Gamma)$, and a term $a \in \Tm(\Delta, A[f])$, there should be a morphism $\langle f, a \rangle : \Delta \to \Gamma.A$ satisfying several equations, as found in \cite{Hofmann_1997}. In the realisability model, this is defined as $\langle f, a \rangle (\delta) = (f_\delta, a_\delta)$. Note that if $d$ tracks $f$ and $e$ tracks $M$, then $\langle d, e \rangle$ tracks $\langle f, a \rangle$, so this is indeed an assembly morphism. If $a \in \Tm(\Gamma, A)$, then we trivially have that $a \in \Tm(\Gamma, A[\id_\Gamma])$, so we get a context morphism $\langle \id_\Gamma, a \rangle : \Gamma \to \Gamma . A$, where $\langle \id_\Gamma, a \rangle (\gamma) = (\gamma, a_\gamma)$. We will denote this morphism by $\overline{a} : \Gamma \to \Gamma.A$, as it will be useful later.

This completes the construction of a category with families from the category of assemblies, allowing us to interpret the structural rules of our dependent type theory in assemblies. What remains is to give interpretations for the different types and type formers.

\subsection{Universe type}
The universe type is interpreted by giving the set of all PERs the complete realisability relation.
\begin{definition}[Universe type]\label{def:universe}
	For any context $\Gamma$, there is a type $\U \in \Ty(\Gamma)$ defined as $U_\gamma = \nabla \Per$, viewing $\Per$ as the set of all PERs. 
    Moreover, there exists a type $\El \in \Ty(\Gamma.\U)$ defined by $\El_{(\gamma, R)} = \iota(R)$, where $\iota(R)$ is the PER $R$ viewed as an assembly. Given a type $A \in \Tm(\Gamma, \U)$ in the universe, we obtain a type $\El[\overline{A}] \in \Ty(\Gamma)$, which we will denote by $\El A$.
\end{definition}

\subsection{Empty, unit and boolean Types}
The empty-, unit-, and boolean types are non-dependent types, and thus are are interpreted as constant families. We write $\0$ and $\1$ for the initial and terminal PERs respectively. Note that for these types, we implicitly assume an induction principle, which is also not explicitly specified in the type theory, but is easily defined in the interpretation.
\begin{definition}[Empty type]
The code for the empty type $\bot^\U \in \Tm(\Gamma, \U)$ will be interpreted as the initial PER for each $\gamma \in \Gamma$:
\[
\bot^\U(\gamma) = \0.
\]
This definition allows us to semantically define the empty type as $\bot \coloneq  \El (\bot^\U)$, where $\El$ denotes the semantic decoder given in \autoref{def:universe}.
\end{definition}

\begin{definition}[Unit type]
The small unit type $\top^\U \in \Tm(\Gamma, \U)$ will be interpreted as the terminal PER for each $\gamma \in \Gamma$:
\[
\top^\U(\gamma) = \1.
\]
As with the empty type, the actual unit type is interpreted as $\El (\top^\U)$, i.e. the terminal assembly for each $\gamma$, making the required type equality true in the interpretation.
The unique term $\star \in \Tm(\Gamma, \top)$ is then given by the unique morphism from $\Gamma$ to the terminal assembly.
\end{definition}

\begin{definition}[Boolean type]
We define $\mathsf{Bool}^\U$ as the following PER:
\begin{align*}
    \forall m, n \in \bN \. \langle 0, m \rangle \sim_{\mathsf{Bool}^\U_\gamma} \langle 0, n \rangle, \\
    \forall m, n \in \bN \. \langle 1, m \rangle \sim_{\mathsf{Bool}^\U_\gamma} \langle 1, n \rangle.
\end{align*}
Again, we define the interpretation of the type $\mathsf{Bool}$ as $\El (\mathsf{Bool}^\U)$. The terms $\mathsf{ff, tt} \in \Tm(\Gamma, \mathsf{Bool})$ are defined to be the constant morphisms $\gamma \mapsto \{\langle 0, m \rangle \mid m \in \bN \}$ and $\gamma \mapsto \{\langle 1, m \rangle \mid m \in \bN \}$, realised by $\Lambda k . \langle 0, k \rangle$ and $\Lambda k . \langle 1, k \rangle$ respectively.
\end{definition}

\subsection{Identity type and function extensionality}
\begin{definition}[Identity type]
We interpret the identity type as follows: for any type $A \in \Ty(\Gamma)$ and terms $a, a' \in \Tm(\Gamma, A)$, we have a type $\Id(A, a, a') \in \Ty(\Gamma)$ defined as 
\[
\Id(A, a, a')_\gamma = \begin{cases}
	\1 & \text{ if $a = a'$},\\
	\0 & \text{ if $a \neq a'$}.
\end{cases}
\]
For each term $a \in \Tm(\Gamma, A)$, we have a term $\refl_a \in \Tm(\Gamma, \Id(A, a, a))$ defined, closely following the set-theoretic model, by $\refl_a(\gamma) = *$.
\end{definition}

Note that the realisability model thus validates uniqueness of identity proofs, since the identity types are either empty or a singleton, but we do not add this axiom to the type theory. 
Informally, function extensionality holds in the realisability model since it is set-theoretic, so (dependent) functions are equal when they are pointwise equal.

\subsection{Dependent product}
To interpret both large and small dependent products in assemblies, we are required to define separate large and small type- and term constructors. This is due to the fact that the inclusion of the small dependent product into assemblies is only isomorphic to the large dependent product, not equal.
To satisfy the definitional type equality
\[
\El (\Pi^U x:A.\, B(x)) \equiv \Pi_{x:\El A} (\El B(x)),
\]
a case distinction is the required in the actual interpretation. We only give the definitions for the product type. First we define the large product and its associated terms.

\begin{definition}[Dependent product]
	Given types $A \in \Ty(\Gamma)$ and $B \in \Ty(\Gamma.A)$, the type $\Pi(A, B) \in \Ty(\Gamma)$ consists of all dependent morphisms from $A$ to $B$, defined as follows:
	\begin{align*}
		&\Pi(A, B)_\gamma = \{ f \in \Pi_{a \in A_\gamma}  B_{(\gamma, a)} \mid \exists e \. e \Vdash_{\Pi(A, B)_\gamma} f \},\\
		&e \Vdash_{\Pi(A, B)_\gamma} f \Iff \forall n \in \bN, \forall a \in A_\gamma \. n \Vdash_{A_\gamma} a \implies e \cdot n \Vdash_{B_{(\gamma, x)}} f (a).
	\end{align*}
	Given a term $b \in \Tm(\Gamma.A, B)$, we define a term $\abs^\Pi_{A, B} (b) \in \Tm(\Gamma, \Pi(A, B))$ as
	\[
	\abs^\Pi_{A, B} (b)(\gamma) = \lambda a \in A_\gamma \. b_{(\gamma, a)}.
	\]
	Note that $\abs^\Pi_{A, B}$ is tracked by $\Lambda m . \Lambda z . \Lambda n\. m \cdot \langle z, n \rangle$. 
	Conversely, given a term $f \in \Tm(\Gamma, \Pi(A, B))$ and $a \in \Tm (\Gamma, A)$, we have a term $\app^\Pi_{A, B}(f, a) \in \Tm(\Gamma, B[\overline{a}])$ defined as follows
	\[
	\app^\Pi_{A, B}(f, a)(\gamma)  = f_\gamma(a_\gamma).
	\]
	Note that $\app^\Pi_{A, B}(f, a)$ is tracked by $\Lambda e . \Lambda n . \Lambda z \. (e \cdot z) \cdot (n \cdot z)$.
\end{definition}

We note a useful property of the definition given here, which will allow for impredicativity in the model. To do so, we will call a type $A \in \Ty(\Gamma)$ modest if $A_\gamma$ is modest for every $\gamma \in \Gamma$.
\begin{proposition}
	For any context $\Gamma$, if $B \in \Ty(\Gamma.A)$ is modest, then $\Pi(A, B)$ is modest for any $A \in \Ty(\Gamma)$.
\end{proposition}
\begin{proof}
	It suffices to show that $\Pi(A, B)_\gamma$ is modest. Suppose that we have $e \in \bN$ such that $e \Vdash_{\Pi(A, B)_\gamma} f_1, f_2 \in \Pi_{a \in A_\gamma}  B_{(\gamma, a)}$. Then for any $a \in A_\gamma$, there is some $n \in \bN$ such that $n \Vdash_{A_\gamma} a$ and thus $e \cdot n \Vdash_{B_{(\gamma, x)}} f_1(a), f_2(a)$, which shows that $f_1(a) = f_2(a)$ since $B$ is modest, and thus $f_1 = f_2$, which shows that $\Pi(A, B)$ is modest.
\end{proof}
Note that we do not allow general impredicativity in the theory, so we do not make use of this property in the general case when interpreting large dependent products. However, the universal quantifier is interpreted simply as an impredicative product over $\Size$, which is thus validated by the model.

Now we give analogous definitions for the small impredicative product type.
\begin{definition}[Small Product]
	Given a type $A \in \Ty(\Gamma)$ and a small type $B \in \Tm(\Gamma.A, \U)$, we define the small (impredicative) product $\pi(A, B) \in \Tm(\Gamma, \U)$ as a $\Gamma$-indexed set of PERs, where $\pi (A, B)_\gamma$ is the PER defined as follows:
	\begin{align*}
		d &\sim_{\pi (A, B)_\gamma} e \Iff \\ &\forall a \in A_\gamma, \forall m, n \in \bN \. m, n \Vdash_{A_\gamma} a \implies (d \cdot m) \sim_{B_{(\gamma, a)}} (e \cdot n).
	\end{align*}
\end{definition}
	Note that $\pi(A, B)_\gamma$ relates two codes if they are realisers for the same dependent morphism $f : (a \in A_\gamma) \to B_{(\gamma, a)}$ and is thus isomorphic to the set of those dependent morphisms.
    
    The definitions of the term formers also follow from the previous observation, although we can also give explicit definitions.
\begin{definition}[Small term formers]
	Given a term $p : (\gamma \in \Gamma) \to \iota(\pi(A, B)_\gamma)$ and $a \in \Tm(\Gamma, A)$ we define a term $\app^\pi(p, a) :(\gamma \in \Gamma) \to \iota(B(\gamma, a))$ as:
	\[
	\app_{A, B}^\pi(p, a)(\gamma) = \{ d \cdot n \mid d \in p_\gamma \land n \Vdash_{A_\gamma} a_\gamma\}.
	\]
	It is clear from the definition of $\pi(A, B)$ that $\app_{A, B}^\pi(p, a)(\gamma)$ is an equivalence class of the PER $B(\gamma, a)$. 
	Conversely, if we are given a term $t : ((\gamma, a) \in \Gamma.A) \to \iota(B(\gamma, a))$, we define a term $\abs^\pi_{A, B}(t) : (\gamma \in \Gamma) \to \pi(A, B)(\gamma)$ as:
	\begin{align*}
		\abs_{A, B}^\pi&(t)(\gamma) = \\
		&\{ e \mid \forall a \in A_\gamma \. \forall n \in \bN \. n \Vdash_{A_\gamma} a \implies (e \cdot n) \in t_{(\gamma, a)} \}.
	\end{align*}
	Equivalently, $\abs_{A, B}^\pi(t)(\gamma)$ is the set of realisers for the dependent morphism defined by \[a \mapsto t_{(\gamma, a)} : (a : A) \to B(\gamma, a)\]
\end{definition}
These definitions are used to interpret small dependent products. More on the formal interpretation of the dependent product and why the case distinction between small and large products is required can be found in \cite{REUS1999128}.

\subsection{Dependent sum}
The interpretation of the dependent sum requires a similar case distinction as the dependent product. We first give the definition of the large dependent sum type, which is simply a generalisation of context extension, defined as follows.
\begin{definition}[Dependent sum]
	Given a type $A \in \Ty(\Gamma)$ and $B \in \Ty(\Gamma.A)$, the type $\Sigma(A, B) \in \Ty(\Gamma)$ is defined as follows:
	\begin{align*}
		& \Sigma(A, B)_\gamma  = \{ (a, b) \mid a \in A_\gamma \land b \in B_{(\gamma, a)}\}, \\
		& \langle n_1, n_2 \rangle \Vdash_{\Sigma(A, B)_\gamma} (a, b) \Iff n_1\Vdash_{A_\gamma} a \and n_2 \Vdash_{B_{(\gamma, a)}} b.
	\end{align*}
	To interpret the pairing operation, we define a morphism \begin{align*}
		\pair^\Sigma_{A, B}&: \Gamma.A.B \to \Gamma.\Sigma(A, B), \\
		\pair^\Sigma_{A, B}&(\gamma, a, b) = (\gamma, (a, b)).
	\end{align*}
	This morphism is tracked by the pairing operation on the realisers.
	To interpret the elimination rule, we are given a type $P \in \Ty(\Gamma.\Sigma(A, B))$ and a term $p \in \Tm(\Gamma.A.B, P[\pair^\Sigma_{A, B}])$ and define a term $\ind_\Sigma^{p} \in \Tm(\Gamma.\Sigma(A, B), P)$ as follows:
	\begin{align*}
		\ind_\Sigma^{ p}(\gamma, (a, b)) = p_{(\gamma, a, b)}.
	\end{align*}
	If $p$ is tracked by $d$, then this term is tracked by $\Lambda k . \Lambda n \. d \cdot \langle k, \pr_1 n, \pr_2 n \rangle$.
\end{definition}
Note that here we are being somewhat informal in defining the realiser for the elimination rule, since we only assumed a computable bijection $\langle \cdot , \cdot \rangle$ for pairing. However, such a computable bijection may be defined for $n$-tuples as well. 

For the dependent sum type, it is easy to see that if $A$ and $B$ are modest, then $\Sigma(A, B)$ is modest too. We give an analogous definition for the type constructor in $\Per$. Note that, unlike for the $\Pi$-type, we do not require separate term constructors in $\Per$.

\begin{definition}[Small dependent sum]
	Given a small type $A \in \Tm(\Gamma, \U)$ and $B \in \Tm(\Gamma. \El A, \U)$, we define the dependent sum $\sigma (A, B) \in \Tm(\Gamma, \U)$ pointwise as the PER $\sim_{\sigma (A, B)(\gamma)}$ defined as follows:
	\begin{align*}
		\langle m_1, n_1 \rangle \sim_{\sigma(A, B)_\gamma} \langle m_2, n_2 \rangle \Iff m_1 \sim_{A(\gamma)} m_2 \and n_1 \sim_{B(\gamma, [m_1]_{A(\gamma)})} n_2.
	\end{align*}
	Note that this definition is equivalent to $\sigma(A, B) = \E(\Sigma (\El A, \El B))$ and thus we have $\sigma(A, B) \cong \Sigma(\El A, \El B)$. 
\end{definition}

\subsection{Type of sizes}
The impredicativity of our realisability model allows us to interpret the type $\Size$ as a large type, while still having universal quantification over $\Size$ with a small type as codomain live in $\U$.  
Thus, given a context $\Gamma$, we define the semantic type $\Size \in \Ty(\Gamma)$ as $\Size_\gamma = \nabla \omega_1$, where $\omega_1$ is the first uncountable ordinal. Note that the set $\omega_1$ is exactly the set of all countable ordinals. Giving the type a full realisability structure encapsulates the idea that sizes should be treated parametrically.

The operators $\circ \in \Tm(\Gamma, \Size)$ and ${\uparrow}  \in \Tm(\Gamma.\Size, \Size)$ can be interpreted as the constant-zero and successor functions respectively:
\begin{align*}
	\circ(\gamma) &= 0, \\
	\uparrow \!(\gamma, \alpha) &= \alpha + 1.
\end{align*}
Note that since $\Size_\gamma$ carries a full realisability structure, these are tracked by any total realiser. 
We will sometimes treat a term $s \in \Tm(\Gamma, \Size)$ of type $\Size$ as a constant value in $\Size_\gamma$, since we build up such terms from the zero, successor and supremum functions, which behave uniformly over $\Gamma$.

The ordering on sizes $\leq^\U$ is interpreted as a small type in a context with two size variables, i.e. $\leq^\U \,\in \Tm(\Gamma.\Size.\Size[\p_\Size], \U)$. The small type $s_1 \leq^\U s_2$ is either the terminal singleton PER or the initial empty PER, depending on whether $s_1 \leq s_2$ holds with the regular relation on natural numbers.
\begin{align*}
	& \le^U (\gamma, s_1, s_2) = 
	\begin{cases}
		\1 & \text{if } s_1 \le s_2,\\
		\0 & \text{otherwise}.
	\end{cases}
\end{align*}
The actual type $\le \, \in \Ty(\Gamma.\Size.\Size[\p_\Size])$ is then simply the lifting to assemblies: $\le \, = \El (\le^\U)$. We will often write $(s_1 \le s_2)_\gamma$ instead of $\le(\gamma, s_1, s_2)$ and we will blur the distinction between $\le$ and $\le^\U$, since the important property is that $s_1 \le s_2$ is pointwise either the initial or terminal object in the category of PERs or assemblies. Note that this interpretation validates the notion that the ordering on sizes in the theory is a proposition.

\subsection{Bounded quantification and the fixpoint operator}
Recall that the notation $\Gamma, i : \Size \vdash \forall j < i \. A(j)$ is shorthand for $\Gamma, i : \Size \vdash \forall j \. (j < i \to A(j))$. Thus, a term $g$ of type $\forall j < i \. A(j)$ is interpreted as a $\Gamma.\Size$-indexed collection of dependent morphisms, where
\[
g_{(\gamma, \alpha)} : (\beta \in \Size_\gamma) \to ((\beta < \alpha)_\gamma \to A_{(\gamma, \beta)}).
\]

If $(\beta < \alpha)_\gamma $ holds in the interpretation, then $(\beta < \alpha)_\gamma$ is the terminal singleton assembly. Then the set of morphisms of type $(\beta < \alpha)_\gamma \to A_{(\gamma, \beta)}$ is isomorphic to $A_{(\gamma, \beta)}$, allowing us to treat $g_{(\gamma, \alpha)}(\beta)$ as an element of $A_{(\gamma, \beta)}$. 
If $(\beta < \alpha)_\gamma$ does not hold in the interpretation, then $(\beta < \alpha)_\gamma$ is the initial (empty) assembly and there thus exists a unique morphism of type $(\beta < \alpha)_\gamma \to A_{(\gamma, \beta)}$, namely the empty morphism, again given by $g_{(\gamma, \alpha)}(\beta)$. This confirms the syntactic intuition that we may treat $g_{(\gamma, \alpha)}$ as a dependent morphism where $\dom(g_{(\gamma, \alpha)}) = \Size_\gamma \upharpoonright \alpha$. We can thus write the type of $g_{(\gamma, \alpha)}$ as:
\[
g_{(\gamma, \alpha)} : (\beta \in \Size_{<\alpha, \gamma}) \to A_{(\gamma, \beta)}.
\]
Given this intuition, we consider the interpretation of the fixpoint operator $\fix$. Given a type $A \in \Ty(\Gamma. \Size)$, and a dependent morphism
\[
f_{\gamma} : (\alpha : \Size_\gamma) \to ((\beta \in \Size_{<\alpha, \gamma}) \to A_{(\gamma, \beta)}) \to A_{(\gamma, \alpha)},
\]
corresponding to a semantic term of the syntactic type $\forall i . (\forall j < i . A(j) \to A(i))$,
we need to define a term $\fix f : (\gamma \in \Gamma) \to \Pi (\Size, A)_\gamma$. Thus, we require that $(\fix f)_\gamma$ is a dependent morphism from $\Size$ to $A$, i.e. 
\[
(\fix f)_\gamma : (\alpha \in \Size_\gamma) \to A_{(\gamma, \alpha)}.
\] 
We define this morphism by well-founded induction on $\Size_\gamma$, i.e. the ordinal $\omega_1$ and thus a well-founded ordering:
\begin{align*}
	(\fix f)_\gamma(\alpha) &= f_\gamma(\alpha, \lambda \beta < \alpha \. (\fix f)_\gamma(\beta)).
\end{align*}
In the case where $\alpha = 0$, there is only the unique empty morphism with type $(\beta \in \Size_{<0, \gamma}) \to A_{(\gamma, \beta)}$, since the set $\Size_{<0, \gamma}$ is empty. 
Note that this definition mirrors exactly the $\beta$-rule for $\fix\!$. The uniqueness property for $\fix f$ also follows by well-founded induction over $\Size_\gamma$. 

For $\fix$ to be a term, it must be tracked by some realiser. Note that, by our definition of $\fix$, if we are given some $f$ realised by $f_r \in \bN$ and $\gamma$ realised by $\gamma_r \in \bN$, we require that a realiser $F$ for $\fix$ satisfying the following for every $n \in \bN$, where we use as convention that the partial operation associates to the left so that we may omit parentheses:
\[
F \cdot f_r \cdot \gamma_r \cdot n = f_r \cdot \gamma_r \cdot n \cdot (\Lambda m \. F \cdot f_r \cdot \gamma_r \cdot m).
\]
In order to define such a realiser, we make use of the fact that a fixpoint operator exists in every partial combinatory algebra:
\begin{proposition}[Fixpoint operator in pca]
    For every pca $(\cA, \cdot)$, there exists a $\fix\! \in \cA$ such that 
    \[
    \fix f \downarrow \text{ and } 
    \fix f \, a \simeq f \, (\fix f) \, a.
    \]
\end{proposition}
A proof of this can be found in \cite{oosten}.
Importantly, it allows us to define a realiser for the fixpoint operator. First, we define a code $\varphi \in \bN$ by
\[
\varphi \coloneq  \Lambda f. \Lambda e . \Lambda k . \Lambda n \. e \cdot k \cdot n \cdot (\Lambda m \. f \cdot e \cdot k \cdot m).
\]
This gives us the following:
\begin{align*}
(\fix \varphi) \cdot f_r \cdot \gamma_r \cdot n &\simeq \, \varphi \cdot (\fix \varphi) \cdot f_r \cdot \gamma_r \cdot n \\
&\simeq f_r \cdot \gamma_r \cdot n \cdot (\Lambda m . (\fix \varphi) \cdot f_r \cdot \gamma_r \cdot m).
\end{align*}
This shows that $\fix \, \varphi$ is a realiser for our fixpoint operator.

\subsection{Existential type}
To interpret the existential type $\exists i . A(i)$, we quotient the corresponding $\Sigma$-type in order for it to live inside the universe $\U$, i.e. canonically construct a partial equivalence relation from the assembly $\Sigma(\Size, A)$. To do so, we make use of the following lemma.

\begin{lemma}
The inclusion functor $I : \Mod \to \Asm$ has a left adjoint $\M : \Asm \to \Mod$.
\end{lemma}
\begin{proof}
Let $A \in \Asm$ be an assembly, we define $\M(A) = ( A {/ \!\sim_\M}, \Vdash_{\M(A)})$ where $\sim_\M$ is the least equivalence relation such that the following holds:
\[
\exists n \in \bN \. n \Vdash_A a_1, a_2 \implies a_1 \sim_\M a_2.
\]
The realisability relation is defined as
\[
n \Vdash_{\M(A)} [ a ]_{\sim_\M} \Iff \exists a' \in [ a ]_{\sim_\M} (n \Vdash_A a').
\]
It is clear that $\M(A)$ is indeed a modest set. We define a map $\eta_A : A \to \M(A)$ as $a \mapsto [a]_{\sim_\M}$, which is realised by $\Lambda x . x$ and is thus a morphism of assemblies. To show that $\M$ extends to a left adjoint of $I$, let $X \in \Mod$ be a modest set and suppose there exists a morphism $g : A \to X$. Then $g$ uniquely factors as $f \circ \eta_A$, where $f : \M(A) \to X$ is defined as $f([a]_{\sim_\M}) = g(a)$. 

It remains to show that $f$ is well-defined. Suppose that $a_1 \sim_\M a_2$; if there exists $n \in \bN$ such that $n \Vdash_A a_1, a_2$, then since $g$ is tracked by some $e \in \bN$, this means that $e \cdot n \Vdash_X g(a_1), g(a_2)$. Thus, since $X$ is a modest set, $g(a_1) = g(a_2)$ and so $f$ is well-defined.

Else, there exist $n_1, ..., n_k \in \bN$ and $b_1, ..., b_{k-1} \in A$ such that $n_1 \Vdash_A a_1, b_1$, $n_2 \Vdash_A b_1, b_2$, ..., and $n_k \Vdash_A b_{k-1}, a_2$. Then by similar reasoning via transitivity, we will get that $g(a_1) = g(b_1) = \dots = g(b_{k-1}) = g(a_2)$.
\end{proof}

Given this construction, we can truncate any assembly $A \in \Asm$ to $\M(A)$ in order (almost) fit in the universe $\U$. However, the universe $\U$ consists of PERs rather than modest sets. Fortunately, the categories $\Per$ and $\Mod$ are equivalent, allowing us define the \emph{$U$-truncation} $\|A\| \in \Tm(\Gamma, \U)$ of a type $A \in \Ty(\Gamma)$ as 
\[\Tr{A}_\gamma = \M(A_\gamma).\]

The $U$-truncation operation should also operate on the level of terms, that is, given a term $a \in \Tm(\Gamma, A)$, there should be a term $\tr{a}$ in the $U$-truncated type. Thus, we define a context morphism 
\begin{align*}
	&\tr{\cdot} : \Gamma.A \to \Gamma.\El (\Tr A), \\
	&\tr{(\gamma, a)} = (\gamma, [a]_{\sim_\M}),
\end{align*}
which is realised by the identity $\Lambda x . x$.
Moreover, the existential type comes with an elimination rule for small types. Thus, a truncated type $\Tr A$ in the model should come with such a rule as well. That is, given a small type $P \in \Tm(\Gamma.\El \Tr A, \U)$ and a term $t \in \Tm(\Gamma.A, (\El P)[\tr \cdot ])$, we define the following term
\begin{align*}
	&\elim_t \in \Tm (\Gamma.\El (\Tr A), \El P), \\
	&\elim_t (\gamma, [a]_{\sim_\M}) = t(\gamma, a).
\end{align*}
To show that $\elim_t$ is well defined, we need to show that $a_1 \sim_M a_2$ implies $t(\gamma, a_1) = t(\gamma, a_2)$. Since $a_1 \sim_M a_2$, we can assume that there exists some $k \in \bN$ such that $k \Vdash_A a_1, a_2$. Moreover, $t$ is tracked by some $e \in \bN$, meaning that the following holds:
\[
\langle m, n \rangle \Vdash_{\Gamma.A} (\gamma, a) \implies e \cdot \langle m, n \rangle \Vdash_{P(\gamma, [a]_{\sim_\M})} t(\gamma, a).
\]
Then for any $\gamma \in \Gamma$, which must realised by some $m \Vdash_\Gamma \gamma$, we have that $\langle m, k \rangle \Vdash_{\Gamma.A} (\gamma, a_1), (\gamma, a_2)$ and thus:
\[
e \cdot \langle m, k \rangle \Vdash_{P(\gamma, [a_1]_{\sim_M})} t(\gamma, a_1), t(\gamma, a_2), 
\]
since $[a_1]_{\sim_M} = [a_2]_{\sim_M}$. Since $P$ is a small type, $P(\gamma, [a_1]_{\sim_M})$ is a modest set when viewed as an assembly. This implies that $t(\gamma, a_1) = t(\gamma, a_2) $ since any realiser codes at most one object. 

Now we are ready to interpret the existential type by truncating the $\Sigma$-type.
\begin{definition}[Existential Type]
	Given a small type $A \in \Tm(\Gamma.\Size, \U)$, the small existential type $\exists i . A(i) \in \Tm(\Gamma, \U)$ is defined as 
	\[\exists i . A(i) \coloneq  \Tr{\Sigma(\Size, \El A)}.\]
	
	As with dependent sums, the pairing operation is interpreted with a morphism $\pair_\exists : \Gamma.\Size.A \to \Gamma. \El (\exists i.A(i))$ defined as 
	\[
	\pair_\exists(\gamma, s, a) = \tr{\pair_\Sigma(\gamma, s, a)},
	\]
	where $\pair_\Sigma(\gamma, s, a) \in \Gamma.\Sigma(\Size, A)$.

	Finally, the elimination rule for existential types is interpreted as follows. Given the following terms:
	\begin{align*}
		P &\in \Tm(\Gamma. \El (\exists i.A(i)), \U),\\
		p &\in \Tm(\Gamma.\Size.A, (\El P)\{\pair_\exists\}),
	\end{align*}
    we define the term $\ind_\exists^{p} \in \Tm(\Gamma.(\exists i. A(i)), \El P)$ as follows:
	\[
		\ind_\exists^{p} = \elim_{\ind_{\Sigma}^{P}} \in \Tm(\Gamma.(\exists i. A(i)), \El P).
	\]
	where $ \elim_{\ind_{\Sigma}^{P}}(\gamma, [(s, a)]_{\sim_\M}) = \ind_{\Sigma}^{P}(\gamma, (s, a))$ is well-defined since $\elim$ is well-defined.

\end{definition}

Now that the formal definition is given, we can make an observation about the structure of the existential type. Given a small type $A \in \Tm(\Gamma, \U)$, note that for the modest set $\M (\Sigma(\Size, \El A)_\gamma)$, the equivalence relation $\sim_\M$ is defined exactly by
\begin{align*}
	(s_1, a_1) \sim_\M (s_2, a_2) &\Iff \exists \langle m, n \rangle \in \bN \. \langle m, n \rangle \Vdash_{\Sigma(\Size, \El A)_\gamma} (s_1, a_1), (s_2, a_2) \\
	&\Iff \exists n \in \bN \. n \Vdash_{\El A_{(\gamma, s_1)}} a_1 \and n \Vdash_{\El A_{(\gamma, s_2)}} a_2
\end{align*} 
Thus, the intuition is that two elements $[(s_1, a_1)]_{\sim_\M}$ and $[(s_2, a_2)]_{\sim_\M}$ of $(\exists i . A(i))_\gamma$ are equal if $a_1$ and $a_2$ share a realiser.
\end{document}